\documentclass[11pt,reqno]{amsart}

\usepackage{amsmath, amsthm, amsopn, amssymb, graphicx, enumerate, geometry, afterpage, caption, subcaption, color, textcomp, bbm}

\newtheorem{thm}{Theorem}[section]
\newtheorem{lem}[thm]{Lemma}
\newtheorem{prop}[thm]{Proposition}
\newtheorem{cor}[thm]{Corollary}

\newtheorem*{thm*}{Theorem}

\theoremstyle{definition}

\newtheorem{setting}{Setting}

\theoremstyle{remark}
\newtheorem*{remark}{Remark}

\newcommand{\R}{\mathbb{R}}
\newcommand{\Z}{\mathbb{Z}}

\newcommand{\T}{\mathbb{T}}
\newcommand{\E}{\mathbb{E}}
\renewcommand{\Pr}{\mathbb{P}}

\def\eps{\varepsilon}
\renewcommand\P{\mathbb{P}}

\DeclareMathOperator\zero{\mathbf{0}}

\DeclareMathOperator{\Var}{Var}
\DeclareMathOperator*{\argmin}{arg\, min}
\DeclareMathOperator*{\argmax}{arg\, max}

\DeclareMathOperator{\Hess}{Hess}
\DeclareMathOperator{\Med}{Med}
\newcommand{\leftrm}{\mathrm{left}}
\newcommand{\rightrm}{\mathrm{right}}

\title[Concentration inequalities for log-concave distributions]{Concentration inequalities for log-concave distributions with applications to random surface fluctuations}

\author{Alexander Magazinov}
\address{Alexander Magazinov\hfill\break
    School of Mathematical Sciences,
    Tel Aviv University,
    Israel and
    Skoltech, Moscow, Russia.}
\email{magazinov-al@yandex.ru}

\author{Ron Peled}
\address{Ron Peled\hfill\break
    School of Mathematical Sciences,
    Tel Aviv University,
    Israel.}
\email{peledron@tauex.tau.ac.il}

\begin{document}
\begin{abstract}
  We derive two concentration inequalities for linear functions of log-concave distributions: an enhanced version of the classical Brascamp--Lieb concentration inequality, and an inequality quantifying log-concavity of marginals in a manner suitable for obtaining variance and tail probability bounds.

  These inequalities are applied to the statistical mechanics problem of estimating the fluctuations of random surfaces of the $\nabla\varphi$ type. The classical Brascamp--Lieb inequality bounds the fluctuations whenever the interaction potential is uniformly convex. We extend these bounds to the case of convex potentials whose second derivative vanishes only on a zero measure set, when the underlying graph is a $d$-dimensional discrete torus. The result applies, in particular, to potentials of the form $U(x)=|x|^p$ with $p>1$ and answers a question discussed by Brascamp--Lieb--Lebowitz (1975). Additionally, new tail probability bounds are obtained for the family of potentials $U(x) = |x|^p+x^2$, $p>2$. This result answers a question mentioned by Deuschel and Giacomin (2000).
\end{abstract}
\maketitle

\section{Introduction}\label{sec:introduction}
The goal of this paper is two-fold: to present new concentration inequalities for log-concave distributions and to apply these inequalities to the problem of bounding the fluctuations of random surfaces of $\nabla\varphi$ type arising in statistical mechanics.

\subsection{Concentration inequalities for log-concave distributions}\label{sec:concentration inequalities}
A log-concave distribution is a probability distribution on $\R^n$ with a density $\exp(-f)$ with respect to Lebesgue measure where $f:\R^n\to(-\infty,\infty]$ is convex (so that $\exp(-f)$ is log-concave). Without restricting generality, we impose the convention that $\{x\colon f(x)<\infty\}$ is open for convex $f$.

Let $X$ be a random vector in $\R^n$ with a log-concave distribution. We wish to study the concentration properties of linear functions $\langle \eta, X\rangle$, for vectors $\eta\in\R^n$, providing both variance and tail probability estimates. A seminal result in this context is the Brascamp--Lieb concentration inequality~\cite{BL75, BL76}, which states that
\begin{equation}\label{eq:Brascamp-Lieb concentration inequality}
  \Var(\langle \eta, X\rangle)\le \E(\langle \eta, (\Hess f)(X)^{-1} \eta\rangle),
\end{equation}
where $\Hess f$ is the Hessian matrix of $f$. In this paper, the matrix $(\Hess f)(x)$ is defined as the unique symmetric matrix for which the second-order Taylor expansion
\begin{equation}\label{eq:second order Taylor intro}
f(x + y) = f(x) +
\langle y, (\nabla f)(x) \rangle +
\frac{1}{2} \langle y, (\Hess f)(x) y \rangle
+ o(\| y \|^2)
\end{equation}
is valid at $x$ for some vector $(\nabla f)(x)$. The convexity of $f$ implies, by Aleksandrov's theorem~\cite{A39} (see also~\cite[Theorem 6.9]{EG15}), that $\Hess f$ exists and is positive semidefinite almost everywhere on $\{x\colon f(x)<\infty\}$. Rockafellar~\cite{R99} discusses the relation of this definition with other possible definitions of second derivatives of $f$. For the one-dimensional case see also Section~\ref{sec:one dimensional log-concave distributions}.
We remark that the expression $\langle \eta, (\Hess f)(X)^{-1} \eta\rangle$ appearing in~\eqref{eq:Brascamp-Lieb concentration inequality} may have a well-defined finite value even if $\Hess f$ is not invertible. This is the case exactly when the kernel of $\Hess f$ is orthogonal to $\eta$; see Lemma~\ref{lem:app:one_point}.

The theorem of Brascamp--Lieb is in fact more general
than~\eqref{eq:Brascamp-Lieb concentration inequality},
allowing to bound also the variance of non-linear functions of $X$ but our focus here is on the linear case. Our first theorem provides an enhancement of the concentration inequality~\eqref{eq:Brascamp-Lieb concentration inequality} which is useful in cases when the expectation of $\langle \eta, (\Hess f)(X)^{-1} \eta\rangle$ is much larger than its typical value (including cases with infinite expectation).

\begin{thm}[Quantile Brascamp--Lieb type inequality]
\label{thm:quantile Brascamp-Lieb}
  There exists a constant $C>0$ so that the following holds. Let $X$ be a random vector with a log-concave density $\exp(-f)$ and let $\eta$ be a non-zero vector in $\R^n$. For each $t>0$,
  \begin{equation}\label{eq:quantile Brascamp Lieb conclusion}
    \Var(\langle \eta, X\rangle)\le \frac{Ct}{\P(\langle \eta, (\Hess f)(X)^{-1} \eta\rangle\le t)^{3}}.
  \end{equation}
  In particular,
  \begin{equation}
    \Var(\langle \eta, X\rangle)\le 8C\Med(\langle \eta, (\Hess f)(X)^{-1} \eta\rangle)
  \end{equation}
  where $\Med(Y)$ is any median of the random variable $Y$, i.e., a number $t$ satisfying $\P(Y\ge t)\ge \frac{1}{2}$ and $\P(Y\le t)\ge \frac{1}{2}$.
\end{thm}
For the follow-up, recall that the classical Pr\'ekopa--Leindler inequality (see below) implies that the density function of $\langle \eta, X\rangle$ is itself log-concave. Our second result asserts a quantitative version of this log-concavity.

\begin{thm}[Quantitative log-concavity]
\label{thm:quantitative log concavity}
Let $X$ be a random vector with a log-concave density $\exp(-f)$ and let $\eta$ be a non-zero vector in $\R^n$. Denote by $\alpha_\eta:\R\to[0,\infty)$ the (log-concave) density function of $\langle \eta, X\rangle$. Define, for each $t>0$ and each $x\in\R^n$ satisfying $f(x)<\infty$,
\begin{equation}\label{eq:E_f_2_def}
D_{\eta, x}(t) := \inf_{\substack{x^+,\, x^-\, \in\, \R^n\\
x^+\, +\, x^-\, =\, 2 x \\
\langle \eta,\, x^+\, -\, x^- \rangle\, =\,2t}}
\frac{f(x^+) +
f(x^-) -
2f(x)}{2}.
\end{equation}
Define further, for each $D\ge0$, $t>0$ and $s\in\R$ satisfying that $\alpha_\eta(s)>0$
\begin{equation}\label{eq:gamma_E_s}
\gamma_{\eta}(D, s, t) := \Pr\big( D_{\eta, X}(t) \ge D \mathbin{\vert}
\langle \eta, X \rangle = s\big).
\end{equation}
Then the inequality
\begin{equation}\label{eq:alpha_f_2_bound}
\sqrt{\alpha_{\eta}(s-t)\alpha_{\eta}(s+t)} \le
\bigl(1-\gamma_{\eta}(D, s, t) (1- e^{-D}) \bigr) \cdot \alpha_{\eta}(s)
\end{equation}
holds for every $D,t$ and $s$ as above. In particular, if $\P(D_{\eta, X}(t) \ge D)\ge \frac{3}{4}$ then
\begin{equation}\label{eq:Markov consequence}
\P\left( \sqrt{\alpha(\langle \eta, X \rangle + t)
 \alpha(\langle \eta, X \rangle - t)}
\leq \Big(1 - \frac{1}{2}\left(1-e^{-D}\right)\Big) \alpha(\langle \eta, X \rangle) \right) \ge \frac{1}{2}.
\end{equation}
\end{thm}

We remark that although $\gamma_{\eta}(D,s,t)$ is defined in~\eqref{eq:gamma_E_s} as a conditional expectation and hence it a priori only makes sense for almost every $s\in\R$ with $\alpha_\eta(s)>0$, we may in fact define it for all such $s$; see~\eqref{eq:gamma eta representative} below.

The main motivation for Theorem~\ref{thm:quantitative log concavity} is
Lemma~\ref{lem:var_via_logconcavity} below.
It is worthwhile to emphasize
that the assumption~\eqref{eq:quantitative log concavity one dimension}
is also the conclusion of Theorem~\ref{thm:quantitative log concavity}. Taken together, the theorem and lemma provide a second route (an alternative to Theorem~\ref{thm:quantile Brascamp-Lieb}) to obtaining variance bounds for linear functions of random vectors with log-concave densities.

\begin{lem}\label{lem:var_via_logconcavity}
There exists a constant $C > 0$ so that the
following holds. Let $\xi$ be a random variable in $\R$ with log-concave density
$\alpha: \mathbb R \to [0, \infty)$. Let $t>0$ and $0<\delta<1$. If
\begin{equation}\label{eq:quantitative log concavity one dimension}
  \P\left( \sqrt{\alpha(\xi + t) \alpha(\xi - t)} \leq (1 - \delta) \alpha(\xi) \right) \ge \frac{1}{2}
\end{equation}
then
\begin{equation}\label{eq:variance_from_1d_logconcavity}
\Var \xi \leq \left(\frac{C t}{\delta}\right)^2.
\end{equation}
\end{lem}

In addition to its use in bounding the variance (via Lemma~\ref{lem:var_via_logconcavity}), Theorem~\ref{thm:quantitative log concavity} can also provide tail probability bounds for the distribution of $\langle \eta, X\rangle$. This is enabled by taking the parameter $t$ large and finding values of $D$ for which $\gamma_{\eta}(D, s, t) (1- e^{-D})$ is close to one for a suitable range of $s$. In our application we do not explore this direction to its fullest and rather demonstrate it in a simpler situation where one has $\gamma_{\eta}(D, s, t)=1$ for suitable $D,s,t$ (see the proof of Theorem~\ref{thm:tail_estimate}). Still, new results are obtained even in this simplified setup.

The proofs of both Theorem~\ref{thm:quantile Brascamp-Lieb} and Theorem~\ref{thm:quantitative log concavity} make use of the Pr\'ekopa--Leindler inequality, introduced in~\cite{P71,L72,P73}, from convexity theory.
\begin{thm*}[Pr\'ekopa--Leindler inequality]\label{thm:prekopa-Leindler}
  Let $m\ge 1$ be an integer, $0<\lambda<1$ be real and $F_1,F_2,F$ be non-negative, Lebesgue integrable real functions on $\R^m$ satisfying
  \begin{equation}\label{eq:Prekopa-Leindler_condition}
    F((1-\lambda)\mathbf x+\lambda \mathbf y)
    \ge F_1(\mathbf x)^{1-\lambda}F_2(\mathbf y)^\lambda,
    \quad \mathbf x, \mathbf y\in\R^m.
  \end{equation}
  Then
  \begin{equation}\label{eq:Prekopa-Leindler_conclusion}
    \int F(\mathbf x)d \mathbf x
    \ge \left(\int F_1(\mathbf x)d \mathbf x\right)^{1-\lambda}
    \left(\int F_2(\mathbf x)d \mathbf x\right)^{\lambda}.
  \end{equation}
\end{thm*}
A proof may be found, for instance, in Schneider's book~\cite[Theorem 7.1.2]{S14} (proving the $m=1$ case directly with a change of variable and proceeding by induction on $m$). In proving Theorem~\ref{thm:quantitative log concavity} we apply the Pr\'ekopa--Leindler inequality with $F_1, F_2$ being the restrictions of the log-concave density $\exp(-f)$ to the hyperplanes $\langle\eta, x\rangle=s+t$ and $\langle\eta, x\rangle=s-t$, respectively (so that $m = n-1$) and set $F(x)=\exp(-f(x) - D_{\eta,x}(t))$ restricted to $\langle\eta, x\rangle=s$. The function $D_{\eta, x}(t)$ is natural in this context as it makes $F$ the smallest function for which the assumption~\eqref{eq:Prekopa-Leindler_condition} is still satisfied with our choices of $F_1, F_2$ and $\lambda = \frac{1}{2}$. The proof of Theorem~\ref{thm:quantile Brascamp-Lieb} makes use of the Pr\'ekopa--Leindler inequality in a similar manner but focuses on an infinitesimal analog of the quantity $D_{\eta, x}(t)$ as $t$ tends to $0$.

\subsection{Fluctuations of random surfaces}\label{sec:introduction fluctuations}
As an application of the above concentration inequalities, we provide new upper bounds for the fluctuations of random surface models of the $\nabla\varphi$ type (see Section~\ref{subs:discussion} for background). We consider the following family of models: Let $G=(V(G),E(G))$ be a finite, connected graph with a distinguished set $\varnothing\neq V_0\subsetneq V(G)$, let $\varphi_0:V_0\to\R$ be a function and let $U$ be a potential, i.e, a measurable function $U:\R\to(-\infty,\infty]$
satisfying $U(x)=U(-x)$ for all $x$. The \emph{random surface model on $G$ with potential
$U$ and boundary conditions $(V_0, \varphi_0)$} is the probability
measure $\mu_{G,V_0,\varphi_0, U}$ on functions
$\varphi:G\to\R$ defined by
\begin{equation}\label{eq:mu_T_n_2_U_measure_def}
  d\mu_{G, V_0,\varphi_0, U}(\varphi) := \frac{1}{Z_{G,V_0,\varphi_0, U}} \exp\Bigg(-\sum_{e\in E(G)}
  U(\nabla_e\varphi)\Bigg) \prod_{v\in V_0}
  \delta\bigl(\varphi(v) - \varphi_0(v)\bigr)
  \prod_{v\in V(G)\setminus V_0} d\varphi(v),
\end{equation}
where $d\varphi(v)$ denotes Lebesgue measure on $\varphi(v)$ and
$\delta(\cdot)$ is the Dirac delta symbol. Thus $\mu_{G, V_0,\varphi_0, U}$ is supported on
\begin{equation}\label{eq:omega_def}
  \Omega_{G, V_0, \varphi_0} :=
\{\text{$\varphi:V(G) \to \mathbb R$ such that $\varphi \vert_{V_0} \equiv \varphi_0$} \},
\end{equation}
and $\prod\limits_{v\in V_0}\delta\bigl(\varphi(v) - \varphi_0(v)\bigr)\prod\limits_{v\in V(G)\setminus V_0} d\varphi(v)$ is Lebesgue measure on this set. The normalization constant (partition function)
\begin{equation}\label{eq:z_def}
  Z_{G,V_0,\varphi_0, U}:=\int_{\R^{V(G)}}\exp\Bigg(-\sum_{e\in E(G)}
  U(\nabla_e\varphi)\Bigg)
  \prod_{v\in V_0}\delta \bigl(\varphi(v) - \varphi_0(v)\bigr)
  \prod_{v\in V(G)\setminus V_0} d\varphi(v),
\end{equation}
makes $\mu_{G, V_0,\varphi_0, U}$ a probability measure. We require that $0<Z_{G,V_0,\varphi_0, U}<\infty$, a requirement which is satisfied under the assumptions of our main results. It is convenient to assume that the edges $E(G)$ are endowed with some fixed orientation, so that $\nabla_e \varphi$ can be unambiguously defined as $\varphi(u) - \varphi(v)$ for an edge $e=(u,v)$, but it is important to note that the choice of orientation is immaterial for expressions such as $U(\nabla_e\varphi)$ as $U$ is an even function.

Our main concern is with the fluctuations of the random surface model on $d$-dimensional lattice graphs with zero boundary values. Specifically, we work with the following two families of graphs:
\begin{setting}\label{ex:torus}
The $d$-dimensional bipartite torus $\T_{2L}^d$: Here $V(\T_{2L}^d):=\{-L+1,-L+2,\ldots,L-1,L\}^{d}$ with $v$ adjacent to $w$ when $v$ and $w$ are equal in all but one coordinate and differ by exactly one modulo $2L$ in that coordinate. We set $V_0 = \{\zero\}$ (where $\zero=(0,\ldots,0)$) and $\varphi_0\equiv 0$. For brevity, we write $\mu_{\T_{2L}^d, U}$ for $\mu_{\T_{2L}^d, V_0, \varphi_0, U}$. As usual, we write $\|v\|_1$ for the $\ell^1$-norm of $v$.
\end{setting}
\begin{setting}\label{ex:box}
The $d$-dimensional box $\Lambda_L^d$: Here $V(\Lambda_L^d):=\{1,\ldots, L\}^d\subset\Z^d$ with the usual nearest-neighbor adjacency of $\Z^d$. We set $V_0$ to be the vertices of $V(\Lambda_L^d)$ which are adjacent in $\Z^d$ to a vertex outside $V(\Lambda_L^d)$ and $\varphi_0\equiv 0$. Again, we write $\mu_{\Lambda_L^d, U}$ for $\mu_{\Lambda_L^d, V_0, \varphi_0, U}$.
\end{setting}

The properties of random surfaces with general interaction potentials first received rigorous consideration in the seminal 1975 work of Brascamp--Lieb--Lebowitz~\cite{BLL75}.
They conjectured that the fluctuations of such random surfaces in dimension $d=3$ are uniformly bounded in the system size for potentials $U$ satisfying that $\int \exp(-p U(x))dx<\infty$ for all $p>0$. Among the main results of their work is a proof, using the Brascamp--Lieb concentration inequality~\eqref{eq:Brascamp-Lieb concentration inequality}, that the conjecture holds for uniformly convex potentials $U$ (i.e., potentials satisfying $\inf_x U''(x)>0$). To date, this remains the main case for which the conjecture is verified (Section~\ref{subs:discussion} details additional progress) with the result missing even for the potential $U(x)=x^4$ which was emphasized in~\cite{BLL75} and in the more recent survey of Velenik~\cite[Problem 1]{V06}.
Our first result verifies the conjecture for a large class of potentials, including the case $U(x) = x^4$ and more generally the family $U(x) = |x|^p$, $1<p<\infty$.
Our result further applies in dimension $d=2$ where it shows that fluctuations are of order at most the square root of the logarithm of the system size,
matching lower bounds of~\cite{BLL75} up to a constant prefactor.

\begin{thm}\label{thm:main}
Suppose that $U:\R\to(-\infty,\infty]$ is such that
$U(x)=U(-x)$ for all $x$, and, in addition, the following assumption
is satisfied:
\begin{equation}\label{eq:main_assumption}
  \text{$U$ is convex and $U''(x)>0$ Lebesgue almost-everywhere (a.e.) on $\{x\colon U(x)<\infty\}$}.
\end{equation}
Let $d\ge 2$ and $L\ge 2$ be integers and let $\varphi$ be
randomly sampled from $\mu_{\T_{2L}^d, U}$ (Setting~\ref{ex:torus}). Then there exists $C>0$, depending on $U$ and $d$ but not on $L$,
such that for any $v\in V(\T_{2L}^d)\setminus\{\zero\}$ we have
\begin{equation}\label{eq:main_var}
\begin{aligned}
  d = 2: \quad & \Var(\varphi(v)) \le C \log(1+\|v\|_1),\\
  d \ge 3: \quad & \Var(\varphi(v)) \le C.
\end{aligned}
\end{equation}
\end{thm}
As discussed above, convexity of $U$ implies the existence of its second derivative almost everywhere on~$\{x\colon U(x)<\infty\}$. Theorem~\ref{thm:main} may be proved using either Theorem~\ref{thm:quantile Brascamp-Lieb} or Theorem~\ref{thm:quantitative log concavity}, and relying on additional facts specific to the random surface model. We demonstrate both approaches in Section~\ref{sec:proof of main theorem}.

As mentioned earlier, linear functions of log-concave distributions have themselves a log-concave distribution. Since log-concave distributions have tail probabilities decaying at least exponentially fast on the scale of their standard deviation, we may deduce such tail probability bounds for the variables $\varphi(v)$ from Theorem~\ref{thm:main}. For uniformly convex potentials, exponential decay may be upgraded to sub-Gaussian decay as a consequence of the Brascamp--Lieb concentration inequality (see, e.g., \cite[Theorem 4.9]{F05}). Deuschel--Giacomin~\cite[Remark 2.11]{DG00} discuss the question of whether in dimensions $d\ge 3$ the tail probabilities of $\varphi(v)$ exhibit faster than sub-Gaussian decay (on the scale of the standard deviation of $\varphi(v)$) when the potential grows faster than quadratically at infinity. Our second theorem shows that this is the case for potentials of the form $U(x) = |x|^p+x^2$, $p>2$. The obtained upper bounds on the tail probabilities match, for $d\neq p$, the lower bounds obtained by~\cite{DG00} up to a constant multiple in the exponent (the lower bounds apply to vertices sufficiently separated from the boundary set $V_0$).


\begin{thm}\label{thm:tail_estimate}
Let $d\ge 3$, $L\ge 2$ be integers and $p>2$ real. Set $U(x) = |x|^p + x^2$. Suppose that either
\begin{itemize}
\item[\it (i)] $\varphi$ is randomly sampled from $\mu_{\T_{2L}^d, U}$,
$V_0 \subset V(\T_{2L}^d)$ is as in Setting~\ref{ex:torus} and
$v \in V(\T_{2L}^d) \setminus V_0$;~or
\item[\it (ii)] $\varphi$ is randomly sampled from
$\mu_{\Lambda_{L}^d, U}$, $V_0 \subset V(\Lambda_{L}^d)$ is as in
Setting~\ref{ex:box} and $v \in V(\Lambda_{L}^d) \setminus V_0$.
\end{itemize}
Then there exist $C,c>0$, depending on $p$ and $d$ but not on $L$ or
$v$, such that for all $t > 2$,
\begin{equation}\label{eq:tail probability bound intro}
  \P(|\varphi(v)|>t)\le \begin{cases}
    C\exp(-c\, t^d) & d<p\\[3pt]
    C\exp\left(-c\, \frac{t^d}{(\log t)^{d-1}}\right) & d=p\\[5pt]
    C\exp(-c\, t^p) & d>p
  \end{cases}.
\end{equation}
\end{thm}
Theorem~\ref{thm:tail_estimate} is derived from Theorem~\ref{thm:quantitative log concavity}, via additional facts specific to the random surface model, in Section~\ref{sec:proof of tail estimate theorem}. An extension to non-zero boundary conditions is discussed at the end of Section~\ref{sec:proof of tail estimate theorem}.

\subsection{Discussion and Background}\label{subs:discussion}

Concentration inequalities for spaces with uniform convexity estimates (possibly of non-quadratic nature) were established in various contexts; see Gromov--Milman~\cite{GrMi87}, Bobkov--Ledoux~\cite{BoLe00} and Milman--Sodin~\cite{MiSo08}. Such concentration inequalities imply, for instance, tail bounds for Lipschitz functions~\cite[Corollary 4.1]{BoLe00}, reminiscent of the tail bounds of Theorem~\ref{thm:tail_estimate}. E. Milman~\cite{Mi19} has informed us that in an unpublished work (circa 2008) with S. Sodin they have been able to use uniform convexity estimates to prove tail probability bounds of the form $C\exp(-c t^p)$ (similarly to~\eqref{eq:tail probability bound intro}) for potentials of the form $|x|^p$ in dimensions $d>p>2$. In this context we emphasize that the concentration inequalities provided by Theorem~\ref{thm:quantile Brascamp-Lieb} and Theorem~\ref{thm:quantitative log concavity} do not require \emph{uniform} convexity but rather apply when a quantitative convexity assumption (in the direction of the functional $\eta$) is given on a \emph{subset} of the full space. Some improvements of the Brascamp--Lieb inequality, which may be used in the absence of uniform convexity, are known in the literature such as the inequality of Bobkov--Ledoux~\cite[Theorem 2.4]{BoLe09} (further extended by Nguyen~\cite{N14} and Kolesnikov--Milman~\cite{KM16, KM17}) and the inequality of Veysseire~\cite{V10} and its extension~\cite[Theorem 4.1]{KM16}.

The random surface models introduced above, sometimes called
{\it $\nabla \varphi$ interface models}, constitute natural examples of
statistical mechanics models with non-compact state space and also serve as effective models for various interfaces arising in spin systems~\cite{G01,F05, V06}. We consider them in the standard setting of the lattice graphs $\mathbb T_{2L}^d$ or $\Lambda_L^d$ and focus on dimensions $d\ge 2$ as one-dimensional random surfaces are equivalent to random walks and are well understood.

The most well-known example of a random surface is the lattice Gaussian free field (LGFF), which has $U(x) = x^2$. In this
case, the Gaussian distribution of the surface considerably simplifies its analysis. A long-standing prediction
is that many properties of the LGFF are universal, holding for a general class of potentials. We briefly review here the progress made on several aspects of this phenomenon.

The \emph{scaling limit} of the LGFF is the \emph{continuum Gaussian free field} (CGFF), the higher-dimensional analogue of Brownian motion. Analogously to the invariance principle for random walks, it is expected that random surfaces satisfying mild integrability conditions on $U$ converge to the CGFF. The state-of-the-art is, however, very far from this goal. A general convergence result~\cite{BrYa90, NaSp97, GOS01, Miller11} has been proved only when $U$ is twice continuously differentiable and satisfies
\begin{equation}\label{eq:U x 2 like}
0<\inf_x U''(x)\le\sup_x U''(x)<\infty.
\end{equation}
Additionally convergence is proved~\cite{BiSp11,Ye19} when $\exp(-U)$ is an average of Gaussian functions, and for non-convex potentials arising as small perturbations of potentials satisfying~\eqref{eq:U x 2 like}~\cite{CD12, AKM16, Hi16, ABKM19}.

Our concern in this paper is rather with the \emph{thermodynamic limit} of the surface --- the study of individual (or multiple) heights of the surface in the infinite-volume limit (i.e., as $L\to\infty$). The most basic aspect of this study is to decide whether the surface localizes or delocalizes. This is quantified, for instance, by the variance at a vertex $v$ with $\|v\|\approx L$ for a surface sampled in the torus $\T_{2L}^d$ (setting~\ref{ex:torus} above), or by the variance at the origin for a surface sampled in $\Lambda_L^d$ (setting~\ref{ex:box}) above. In the seminal 1975 work of Brascamp--Lieb--Lebowitz~\cite{BLL75} it is conjectured that for every potential $U$ satisfying that $\int \exp(-p U(x))dx<\infty$ for all $p>0$, the variance is uniformly bounded in $L$ when $d\ge 3$ while diverging with $L$ when $d=2$. It is further expected that the divergence in two dimensions occurs at rate $\log(L)$.
Obtaining a uniform upper bound on the variance in $d\ge 3$ is noted as an interesting open problem also by Fr\"ohlich and Pfister~\cite{FP81}, who suggest it for twice continuously differentiable potentials whose second derivative grows at most at a power-law rate.

Arguments of Mermin--Wagner type have successfully established a logarithmic \emph{lower} bound on the variance in two dimensions for a wide class of potentials, including \emph{all} twice-continuously differentiable potentials~\cite{BLL75, DS78, FP81, ISV02, MP15}. In contrast, rigorous \emph{upper} bounds on the variance, which may be viewed as a form of continuous-symmetry breaking, have so far been more scarce: Brascamp--Lieb--Lebowitz~\cite{BLL75} proved upper bounds of the correct order for twice-continuously differentiable potentials satisfying $\inf_x U''(x)>0$, and for quadratic-growth potentials reduced to this case via decimation, using the Brascamp--Lieb variance inequality~\eqref{eq:Brascamp-Lieb concentration inequality}. Upper bounds are also shown for the class of surfaces described after~\eqref{eq:U x 2 like}~\cite{BiSp11, BrSp12, Ye19, CDM09, CD12, AKM16, Hi16, ABKM19}, and for the potential $U(x)=|x|$ in $d\ge 3$ via infra-red bounds~\cite{BFL82}. Despite this remarkable progress, localization has remained open even for the potential $U(x) = x^4$, which was emphasized as an open problem in~\cite{BLL75} and~\cite[Problem 1]{V06}. Theorem~\ref{thm:main} adds significantly to the known cases by establishing localization for the family $U(x) = |x|^p$, $1<p<\infty$, and more generally for the class of $U$ satisfying~\eqref{eq:main_assumption}.

To the authors' knowledge, Theorem~\ref{thm:tail_estimate} is the first example in the published literature of a real-valued random surface for which the tail probability $\P(|\varphi(v)|>t)$ has faster than sub-Gaussian decay on the scale of the standard deviation of $\varphi(v)$. The unpublished work of Milman--Sodin mentioned above is an earlier example, applying to the potential $|x|^p$ in dimensions $d>p>2$. This case may also be treated via the methods of Theorem~\ref{thm:tail_estimate}; see the remark at the end of Section~\ref{sec:proof of tail estimate theorem}.

\subsection{Reader's guide}
Our results on general log-concave distributions are derived in Section~\ref{sec:proofs of concentration ienqualities}. The section begins by considering one-dimensional log-concave distributions and proving Lemma~\ref{lem:var_via_logconcavity}, and then proceeds with the proofs of Theorem~\ref{thm:quantile Brascamp-Lieb} (quantile Brascamp-Lieb type inequality) and Theorem~\ref{thm:quantitative log concavity} (quantitative log-concavity). Our results on random surfaces are proved in Section~\ref{sec:random surface preliminaries} and Section~\ref{sec:proofs of results on random surfaces}. Section~\ref{sec:isoperimetry} quantifies the isoperimetric properties of $\Z^d$ and the robustness of their anchored version under (not necessarily independent but very super-critical) bond percolation. Section~\ref{sec:energy} establishes lower bounds on the `generalized energy' $\sum_{e\in E(G)} U(\nabla_e \varphi)$ for general graphs $G$ and even convex $U$ in terms of the isoperimetry of $G$ (in the spirit of Benjamini and Kozma~\cite[Theorem 2.1]{BK05}, but applicable to non-quadratic potentials). Section~\ref{sec:sparsity} bounds the probability that when $\varphi$ is sampled from the random surface measure then $|\nabla_e \varphi|\in S$ for a given small set $S\subset[0,\infty)$ and for all edges $e$ in a given subset. The estimate is proved using the chessboard estimate and is the only place in our proof for which periodic boundary conditions (i.e., Setting~\ref{ex:torus}) are essential. In Section~\ref{sec:proofs of results on random surfaces} we adapt the general concentration results of Section~\ref{sec:concentration inequalities} to the random surface setting, obtaining statements (Lemma~\ref{lem:upper_bound_1}, Lemma~\ref{lem:quantile Brascamp-Lieb random surfaces} and Lemma~\ref{lem:quantitative log concavity random surfaces}) which highlight the role that the above energy estimates play in quantifying the log-concavity of the measure. These are followed by the proofs of Theorem~\ref{thm:main} and Theorem~\ref{thm:tail_estimate}. Of the preliminary results developed in Section~\ref{sec:random surface preliminaries}, only Lemma~\ref{lem:isoperimetry of cube} and Lemma~\ref{lem:energy_estimate_2} are used in the derivation of Theorem~\ref{thm:tail_estimate} from Theorem~\ref{thm:quantitative log concavity}.

\section{Proofs of concentration inequalities}\label{sec:proofs of concentration ienqualities}
In this section we prove the results of Section~\ref{sec:concentration inequalities}: Theorem~\ref{thm:quantile Brascamp-Lieb}, Theorem~\ref{thm:quantitative log concavity} and Lemma~\ref{lem:var_via_logconcavity}.

\subsection{One-dimensional distributions with log-concave density}\label{sec:one dimensional log-concave distributions}

We begin by assembling a collection of useful facts about
one-dimensional log-concave distributions, including a proof of Lemma~\ref{lem:var_via_logconcavity}.

As discussed in Section~\ref{sec:concentration inequalities}, the second derivative of a convex function may be defined using the validity of a Taylor expansion, which in the one-dimensional case takes the form
\begin{equation}\label{eq:second order Taylor one dimension}
  f(x + y) = f(x) + yf'(x) + \frac{1}{2}f''(x)y^2 + o(y^2).
\end{equation}
In fact, in the one-dimensional case the usual definitions of $f'$ and $f''$ may also be used, in the following way: Right and left derivatives, $f'_{\leftrm}(t)$ and $f'_{\rightrm}(t)$, exist everywhere on $I := \{ x \in \R : f(x) < \infty \}$, are non-decreasing and coincide at all but countably many points of $I$. Consequently, $f''$, defined as the derivative of any non-decreasing extension of $f'$ to $I$ (say, $f'_{\rightrm}(t)$), exists almost everywhere on $I$ by Lebesgue's theorem on increasing functions~\cite[Theorem~1 of Section~3.2.2]{Ka18}. It is straightforward that the expansion~\eqref{eq:second order Taylor one dimension} then holds for almost every $x\in I$. In addition, the inequality
\begin{equation}\label{eq:first derivative difference inequality}
  f'_{\rightrm}(t_+) - f'_{\leftrm}(t_-) \ge
\int\limits_{t_-}^{t_+} f''(t) \, dt
\end{equation}
holds for all $t_-<t_+$ in $I$; see~\cite[Theorem~1 of Section~7.2.1]{Ka18}.

The following proposition relates the concentration properties of one-dimensional log-concave distributions to the maximum of their density.
\begin{prop}\label{prop:1d_logconcave_properties}
Let $\xi$ be a real-valued random variable with log-concave distribution. Denote its density by $\alpha: \mathbb R \to [0, \infty)$. Then
\begin{enumerate}
  \item \label{item:logconcave_max}for some absolute constants $0 < c_1 < 1 < C_1$ it holds that
  \begin{equation*}
  \frac{c_1}{\sqrt{\Var \xi}} \leq
  \sup\limits_{s \in \mathbb R} \alpha(s) \leq
  \frac{C_1}{\sqrt{\Var \xi}}\,;
\end{equation*}
  \item \label{item:logconcave_decay}for an absolute constant $C_2>1$ it holds that
  \begin{equation*}
  \sup\limits_{s \in \mathbb R} \alpha(s) \leq
  C_2 \cdot \inf \left\{ a \in \mathbb R : \Pr (\alpha(\xi) < a) >
  \frac{1}{4} \right\}\,;
  \end{equation*}
  \item \label{item:logconcave_tail}for each $p>0$,
  \begin{equation*}
  \Pr( \alpha(\xi) \le p ) \leq
  \frac{p}{\sup\limits_{s \in \mathbb R} \alpha(s)}.
\end{equation*}
\end{enumerate}
\end{prop}
These properties are standard. Proofs may be found, for instance, in ~\cite[Proposition 4.1]{Bob99} (Item~\ref{item:logconcave_max}), ~\cite[Lemma 5.2]{Kla07} (Item~\ref{item:logconcave_decay}) and~\cite[Lemma 5.6]{LV07} (Item~\ref{item:logconcave_tail}).

The next lemma controls the quantiles of the second logarithmic derivative of a log-concave density.
\begin{lem}\label{lem:app:max_density}
Let $\xi$ be a real-valued random variable with log-concave density $\alpha=\exp(-f)$.
Let $C \ge 4$ and let
$M = \max\limits_{t \in \mathbb R} \alpha(t)$. Then
\begin{equation*}
\Pr \left(f''(\xi) > (CM)^2
\right) \leq 4C^{-1}.
\end{equation*}
\end{lem}

\begin{proof}
Choose $t_+ \in \R$ so that $\Pr(\xi > t_+) = C^{-1}$.
We claim that $f'_{\rightrm}(t_+) \leq CM$. Indeed, assume the converse.
In that case, $f'(x) > CM$ almost everywhere on the set
$\{ x \in (t_+, \infty) : f(x) < \infty \}$ and therefore it holds that
$f(t_+ + y) > f(t_+) + CMy$ for every $y > 0$. Given that
$\exp( -f(t_+)) \leq M$, we have
\begin{equation*}
\Pr(\xi > t_+) = \int\limits_{0}^{\infty} \exp(-f(t_+ + y)) \, dy \leq
M \int\limits_{0}^{\infty} e^{-CMy}  \, dy
= \frac{M}{CM} = C^{-1},
\end{equation*}
a contradiction.

Similarly, choose $t_-$ so that $\Pr(\xi < t_-) = C^{-1}$. Just as above, we conclude that $f'_{\leftrm}(t_-) \geq -CM$.

Now we have
\begin{equation*}
2CM \stackrel{(a)}= CM - (-CM) \stackrel{(b)}\geq f'_{\rightrm}(t_+) - f'_{\leftrm}(t_-) \stackrel{(c)}\ge
\int\limits_{t_-}^{t_+} f''(t) \, dt
\stackrel{(d)}\geq (CM)^2 \int\limits_{t_-}^{t_+}
 \mathbbm{1}\{f''(t) \geq (CM)^2 \} \, dt
\end{equation*}
where the inequality (c) is~\eqref{eq:first derivative difference inequality} above. Therefore
\begin{multline*}
\Pr\left(t_- \le \xi \le t_+ \; \text{and} \; f''(\xi) > (CM)^2\right) =
\int\limits_{t_-}^{t_+} \exp(-f(t)) \cdot
\mathbbm{1}\{f''(t) \geq (CM)^2 \} \, dt \\
\leq M \cdot \frac{2CM}{(CM)^2} = 2C^{-1}.
\end{multline*}

In total, we get
\begin{multline*}
\Pr(f''(\xi) > (CM)^2)
\leq \Pr(\xi < t_-) + \Pr(\xi > t_+) +
\Pr(t_- \le \xi \le t_+ \; \text{and} \; f''(\xi) > (CM)^2) \\
\leq C^{-1} + C^{-1} + 2C^{-1} = 4C^{-1},
\end{multline*}
as desired.
\end{proof}

Finally, we proceed with the proof of
Lemma~\ref{lem:var_via_logconcavity}.

\begin{proof}[Proof of Lemma~\ref{lem:var_via_logconcavity}]
Our goal will be to show that the choice
\begin{equation*}
C := 8C_1 \max (1, \ln C_2)
\end{equation*}
suffices, where the values of $C_1$ and $C_2$ are those from Item~\ref{item:logconcave_max} and Item~\ref{item:logconcave_decay} of
Proposition~\ref{prop:1d_logconcave_properties}.
We argue by contradiction, assuming
that~\eqref{eq:variance_from_1d_logconcavity} does not hold with this choice of $C$.

Denote
\begin{equation*}
a_0 := \inf \left\{ a \in \mathbb R : \Pr (\alpha(\xi) < a) >
  \frac{1}{4} \right\}.
\end{equation*}
Due to log-concavity of $\alpha$, there is a unique pair of values
$r_1 < r_2$ such that $\alpha(s) < a_0$ whenever $s < r_1$ or $s > r_2$,
and $\alpha(s) > a_0$ whenever $r_1 < s < r_2$.
By definition of $a_0$, we have
\begin{equation*}
\int\limits_{r_1}^{r_2} \alpha(s)\, ds = \frac{3}{4},
\end{equation*}
and, consequently,
\begin{equation}\label{eq:logconcavity_proof_1}
r_2 - r_1 > \frac{1}{\sup_{s \in \mathbb R} \alpha(s)}
\int\limits_{r_1}^{r_2} \alpha(s)\, ds
= \frac{3}{4 \cdot \sup_{s \in \mathbb R} \alpha(s)}
> \frac{3 \sqrt{\Var \xi}}{4C_1},
\end{equation}
where the last inequality is an application of Item~\ref{item:logconcave_max} of
Proposition~\ref{prop:1d_logconcave_properties}.
By our assumption, the
inequality~\eqref{eq:variance_from_1d_logconcavity} does not hold, i.e.,
$\sqrt{\Var \xi} > \frac{C t}{\delta}$, which, combined
with~\eqref{eq:logconcavity_proof_1}, yields
$r_2 - r_1 > \frac{3Ct}{4C_1 \delta}$.

Denote $k := \frac{C}{8C_1 \delta}$. Then $r_1 + kt < r_2 - kt$, and,
moreover,
\begin{equation*}
  \Pr( \xi \in [r_1 + kt, r_2 - kt] ) \geq \frac{3}{4} - 2kt \cdot \sup\limits_{s \in \mathbb R} \alpha(s) \geq \frac{3}{4} - \frac{2C_1kt}{\sqrt{\Var \xi}}
= \frac{3}{4} - \frac{Ct / \delta}{4\sqrt{\Var \xi}} > \frac{1}{2}.
\end{equation*}
Combined with~\eqref{eq:quantitative log concavity one dimension},
this means that the set
\begin{equation*}
  S = \left \{ s \in (r_1 + kt, r_2 - kt) :
  \sqrt{\alpha(s + t) \alpha(s - t)} \leq (1 - \delta) \alpha(s) \right\}
\end{equation*}
is non-empty. Let $s_0 \in S$. Then either
$\alpha(s_0 + t) \leq (1 - \delta) \alpha(s_0)$, or
$\alpha(s_0 - t) \leq (1 - \delta) \alpha(s_0)$. We consider only the
first case, since the other one is similar. By the choice of $C$,
we have $k > 1$ and therefore $s_0 < s_0 + t < s_0 + kt < r_2$.
Consequently, the log-concavity of $\alpha$ implies that
\begin{equation*}
  a_0 \leq \lim\limits_{r \nearrow r_2} \alpha(r) < (1 - \delta)^k \alpha(s_0).
\end{equation*}
But
\begin{equation*}
  (1 - \delta)^k = \exp( k \ln (1 - \delta)) < \exp(-k \delta) =
  \exp \left(- \frac{C}{8C_1}\right) \leq \frac{1}{C_2},
\end{equation*}
which implies that
\begin{equation*}
a_0 < \frac{\alpha(s_0)}{C_2} \leq
\frac{\sup_{s \in \R} \alpha(s)}{C_2}.
\end{equation*}
This contradicts Item~\ref{item:logconcave_decay} of Proposition~\ref{prop:1d_logconcave_properties} and thereby finishes the proof.
\end{proof}

\subsection{Proof of the quantile Brascamp--Lieb type inequality}
In this section we prove Theorem~\ref{thm:quantile Brascamp-Lieb}.

The next lemma connects the quadratic form $\langle \mathbf n, (\Hess f)(\mathbf x)^{-1} \mathbf n \rangle$ to a quantitative measure of the convexity of $f$ at $\mathbf x$.
\begin{lem}\label{lem:app:one_point}
Let $f : \mathbb R^n \to (-\infty, \infty]$ be a
convex function and $\mathbf x \in \mathbb R^n$ be a point such that
the Taylor expansion
\begin{equation}\label{eq:second order Taylor}
f(\mathbf x + \mathbf y) = f(\mathbf x) +
\langle \mathbf y, \mathbf v \rangle +
\frac{1}{2} \langle \mathbf y, H \mathbf y \rangle
+ o(\| \mathbf y \|^2)
\end{equation}
exists for some vector $\mathbf v$ and positive definite matrix $H$.
Let $\mathbf n \in \mathbb R^n$ be a unit vector and set
$s:=\langle \mathbf n, \mathbf x\rangle$. Then the following holds
as $\gamma \searrow 0$:
\begin{equation}\label{eq:app:one_point}
\inf\limits_{\substack
{\mathbf x^+, \mathbf x^- \in \mathbb R^n : \\
\mathbf x^+ + \mathbf x^- = 2 \mathbf x \\
\langle \mathbf n, \mathbf x^{\pm} \rangle = s \pm \gamma}}
\biggl( f(\mathbf x^+) + f(\mathbf x^-) - 2f(\mathbf x) \biggr) =
\frac{\gamma^2}
{\langle \mathbf n, H^{-1} \mathbf n \rangle}
+ o(\gamma^2).
\end{equation}
\end{lem}

\begin{proof}
Replacing $f$ by $\tilde{f}$, where
$\tilde{f}(\mathbf z) := f(\mathbf z) -
\langle \mathbf z - \mathbf x, \mathbf v \rangle - f(\mathbf x)$
does not change~\eqref{eq:app:one_point}. Hence we do not lose any
generality assuming that $f(\mathbf x) = 0$ and
$\mathbf v = \mathbf 0$. The rest of the argument is provided
under these assumptions.

We claim that the following approximation
holds if $\gamma > 0$ is sufficiently small:
\begin{equation}\label{eq:level_inf}
\inf\limits_{\substack{x^\pm  \in \mathbb R^n \\ \langle \mathbf n, \mathbf x^\pm \rangle = s \pm \gamma}}
f(x^\pm) = \frac{1}{2} \left\langle
\gamma \mathbf y_0, H \gamma \mathbf y_0
\right\rangle + o(\gamma^2) =
\frac{\gamma^2}
{2 \langle \mathbf n, H^{-1} \mathbf n \rangle}
+ o(\gamma^2)\quad\text{where}\quad \mathbf y_0 := \frac{H^{-1} \mathbf n}
{\langle \mathbf n, H^{-1} \mathbf n \rangle}.
\end{equation}
We use the standard fact that if $H$ is an $n\times n$ positive definite matrix and $\mathbf n\in\R^n\setminus\{0\}$ then
\begin{equation}\label{eq:quadratic form variational principle}
    \inf_{\substack{\mathbf y \in \mathbb R^n \\
 \langle \mathbf n, \mathbf y \rangle = 1}}
 \langle \mathbf y, H \mathbf y \rangle = \frac{1}{\langle \mathbf n, H^{-1} \mathbf n \rangle}\quad\text{and}\quad
    \argmin\limits_{\substack{\mathbf y \in \mathbb R^n \\
 \langle \mathbf n, \mathbf y \rangle = 1}}
 \langle \mathbf y, H \mathbf y \rangle = \frac{H^{-1} \mathbf n}{\langle \mathbf n, H^{-1} \mathbf n \rangle}.
\end{equation}
(these extend to the case that $H$ is positive semidefinite, with a suitable interpretation depending on whether $\mathbf n$ is orthogonal to the kernel of $H$ or not).

We will prove~\eqref{eq:level_inf} with the plus sign,
as the other case is identical. Let $R := \| \mathbf y_0 \|$.
The strict convexity of the quadratic form $\mathbf y\mapsto\langle \mathbf y, H\mathbf y\rangle$ implies that for all sufficiently small $\gamma > 0$ one has
\begin{equation*}
\inf\limits_{\substack{x^+  \in \mathbb R^n \\ \langle \mathbf n, \mathbf x^+ \rangle = s + \gamma \\
\| \mathbf x^+ - \mathbf x \| = 2R \gamma}}
f(\mathbf x^+) > f(\mathbf x + \gamma \mathbf y_0),
\end{equation*}
and, consequently, by the convexity of $f$,
\begin{equation}\label{eq:level_inf_outside}
\inf\limits_{\substack{\mathbf x^+  \in \mathbb R^n \\ \langle \mathbf n, \mathbf x^+ \rangle = s + \gamma \\
\| \mathbf x^+ - \mathbf x \| \geq 2R \gamma}}
f(\mathbf x^+) > f(\mathbf x + \gamma \mathbf y_0) =
\frac{1}{2}\langle \gamma \mathbf y_0, H \gamma \mathbf y_0 \rangle + o(\gamma^2).
\end{equation}
On the other hand,
\begin{equation}\label{eq:level_inf_inside}
\inf\limits_{\substack{\mathbf x^+  \in \mathbb R^n \\ \langle \mathbf n, \mathbf x^+ \rangle = s + \gamma \\
\| \mathbf x^+ - \mathbf x \| \leq 2R \gamma}}
f(\mathbf x^+) = \frac{1}{2}\inf\limits_{\substack{\mathbf y \in \mathbb R^n \\ \langle \mathbf n, \mathbf y \rangle = \gamma}}
\langle \mathbf y, H \mathbf y \rangle + o(\gamma^2) =
\frac{1}{2}\langle \gamma \mathbf y_0, H \gamma \mathbf y_0 \rangle + o(\gamma^2).
\end{equation}
Combining~\eqref{eq:level_inf_outside} and~\eqref{eq:level_inf_inside}
indeed yields~\eqref{eq:level_inf} with the plus sign.

Now the lower bound for the left-hand side
of~\eqref{eq:app:one_point} is obtained as follows:
\begin{multline*}
\inf\limits_{\substack
{\mathbf x^+, \mathbf x^- \in \mathbb R^n : \\
\mathbf x^+ + \mathbf x^- = 2 \mathbf x \\
\langle \mathbf n, \mathbf x^{\pm} \rangle = s \pm \gamma}}
\biggl( f(\mathbf x^+) + f(\mathbf x^-) - 2f(\mathbf x) \biggr) =
\inf\limits_{\substack
{\mathbf x^+, \mathbf x^- \in \mathbb R^n : \\
\mathbf x^+ + \mathbf x^- = 2 \mathbf x \\
\langle \mathbf n, \mathbf x^{\pm} \rangle = s \pm \gamma}}
\biggl( f(\mathbf x^+) + f(\mathbf x^-) \biggr) \geq \\
\inf\limits_{\substack
{\mathbf x^+ \in \mathbb R^n : \\
\langle \mathbf n, \mathbf x^+ \rangle = s + \gamma}}
f(\mathbf x^+) +
\inf\limits_{\substack
{\mathbf x^- \in \mathbb R^n : \\
\langle \mathbf n, \mathbf x^- \rangle = s - \gamma}}
f(\mathbf x^-) =
\frac{\gamma^2}
{\langle \mathbf n, H^{-1} \mathbf n \rangle}
+ o(\gamma^2).
\end{multline*}
The matching upper bound is achieved by setting
$\mathbf x^\pm := \mathbf x \pm \gamma \mathbf y_0$.
\end{proof}

Let now $\exp(-f)$ be a log-concave probability density on
$\R^n$, and let $X$ be a random vector in $\R^n$ sampled
according to the probability measure $\exp(-f(\mathbf x))\, d \mathbf x$.
As mentioned above, it is a simple corollary of the Pr\'ekopa--Leindler inequality that the distribution of a one-dimensional marginal
$\langle \mathbf n, X \rangle$, where $\mathbf n$ is a unit vector,
is log-concave. Our next lemma quantifies this log-concavity
by making use of Lemma~\ref{lem:app:one_point}.

From now on, we will write $m_{n-1}$ for the $(n-1)$-dimensional Hausdorff measure on $\R^n$.

\begin{lem}\label{lem:app:second_derivative}
Let $\exp(-f)$ be a log-concave probability density on $\R^n$.
Let $\mathbf n \in \mathbb R^n$ be a unit vector. Define $\alpha:\R\to[0,\infty)$ (the log-concave marginal density) by
\begin{equation*}
\alpha(r) := \int\limits_{\{\mathbf x\in\R^n\colon \langle \mathbf n, \mathbf x\rangle = r\}} \exp(-f(\mathbf x))\,d m_{n-1}(\mathbf x).
\end{equation*}
Let $s\in\R$ be such that $\alpha(s)>0$. Suppose $\Hess f$ is a positive definite matrix $m_{n-1}$-almost-everywhere on $\{\mathbf x\in\R^n\colon \langle \mathbf n, \mathbf x\rangle = s, f(\mathbf x)<\infty\}$. Suppose further that $(\ln \alpha)''$ is defined at $s$ (in the sense of Section~\ref{sec:one dimensional log-concave distributions}). Then
\begin{equation}\label{eq:app:second_derivative}
-(\ln \alpha)''(s) \geq
\frac
{\int\limits_{\{\mathbf x\in\R^n\colon \langle \mathbf n, \mathbf x\rangle = s, f(\mathbf x)<\infty\}}
\frac{1}{\langle \mathbf n, (\Hess f)(\mathbf x)^{-1} \mathbf n \rangle} \cdot
 e^{-f(\mathbf x)} \,d m_{n-1}(\mathbf x)}
{\alpha(s)}.
\end{equation}
\end{lem}

\begin{proof}
Define
\begin{align*}
\Pi & := \{ \mathbf x \in \mathbb R^n :
\langle \mathbf n, \mathbf x \rangle = s, f(\mathbf x)<\infty, (\Hess f)(\mathbf x)\text{ is positive definite} \}, \\
m(\mathbf x, \gamma) & := \inf\limits_{\substack
{\mathbf x^+, \mathbf x^- \in \mathbb R^n : \\
\mathbf x^+ + \mathbf x^- = 2 \mathbf x \\
\langle \mathbf n, \mathbf x^{\pm} \rangle = s \pm \gamma}}
\biggl( f(\mathbf x^+) + f(\mathbf x^-) - 2f(\mathbf x) \biggr), \\
X(\gamma, \varepsilon) & :=
\left\{ \mathbf x \in \Pi :
\text{
$m(\mathbf x, \gamma_1) \geq \frac{(1 - \varepsilon) \gamma_1^2}
{\langle \mathbf n, (\Hess f)(\mathbf x)^{-1} \mathbf n \rangle}$
for all $0 < \gamma_1 < \gamma$}
\right\}.
\end{align*}
By Lemma~\ref{lem:app:one_point}, for every $\varepsilon > 0$ it holds that
$X(\gamma, \eps)$ increases to $\Pi$ as $\gamma$ decreases to $0$.

Fix $\gamma,\varepsilon>0$. For $\gamma_1 \in (0, \gamma)$, we may apply the
Pr\'ekopa--Leindler inequality to the following functions on $\{ \mathbf x \in \mathbb R^n \colon \langle \mathbf n, \mathbf x \rangle = s\}$:
\begin{multline*}
F(\mathbf x) := \exp (-f(\mathbf x)) \cdot p(\mathbf x), \quad
F_1(\mathbf x) := \exp (-f(\mathbf x + \gamma_1 \mathbf n)), \quad
F_2(\mathbf x) := \exp (-f(\mathbf x - \gamma_1 \mathbf n)), \\
\text{where} \qquad p(\mathbf x) := \left \{
\begin{array}{ll}
\exp \left( - \frac{(1 - \varepsilon) \gamma_1^2}
{2 \langle \mathbf n, (\Hess f)(\mathbf x)^{-1} \mathbf n \rangle}\right) &
\quad \text{if $\mathbf x \in X(\gamma, \varepsilon)$}, \\
1 &
\quad \text{if $\mathbf x \in \Pi \setminus X(\gamma, \varepsilon)$}
\end{array}
\right.
\end{multline*}
to obtain
\begin{equation}\label{eq:log-concavity first estimate}
\alpha(s + \gamma_1)^{\frac{1}{2}} \alpha(s - \gamma_1)^{\frac{1}{2}}
\leq \alpha(s)\left( 1 - \frac
{\int\limits_{X(\gamma, \varepsilon)}
\left(
1 - \exp \left( - \frac{(1 - \varepsilon) \gamma_1^2}
{2 \langle \mathbf n, (\Hess f)(\mathbf x)^{-1} \mathbf n \rangle}\right)
\right) \cdot
 e^{-f(\mathbf x)}\,d m_{n-1}(\mathbf x)}
{\alpha(s)}
\right).
\end{equation}
Now set
\begin{equation*}
  X(\gamma,\eps,\delta):=\{\mathbf x\in X(\gamma,\eps)\colon \langle \mathbf n, (\Hess f)(\mathbf x)^{-1} \mathbf n \rangle\ge \delta\}
\end{equation*}
and note that $X(\gamma,\eps,\delta)$ increases to $X(\gamma,\eps)$ as $\delta$ decreases to $0$ as $\Hess f$ is positive definite on $\Pi$. It follows that~\eqref{eq:log-concavity first estimate} continues to hold when $X(\gamma,\eps)$ is replaced by $X(\gamma,\eps,\delta)$ for any $\delta>0$. Taking logarithms and using the Taylor expansion (of the logarithm and the exponential) we obtain
\begin{multline*}
\frac{1}{2} \ln \alpha(s + \gamma_1) +
\frac{1}{2} \ln \alpha(s - \gamma_1) - \ln \alpha(s) \leq\\
 \ln\left( 1 - \frac
{\int\limits_{X(\gamma, \varepsilon,\delta)}
\left(
1 - \exp \left( - \frac{(1 - \varepsilon) \gamma_1^2}
{2 \langle \mathbf n, (\Hess f)(\mathbf x)^{-1} \mathbf n \rangle}\right)
\right) \cdot
 e^{-f(\mathbf x)}\,d m_{n-1}(\mathbf x)}
{\alpha(s)}
\right)=\\
- \gamma_1^2 \cdot \frac
{\int\limits_{X(\gamma,\eps,\delta)}
\frac{1 - \varepsilon}
{2 \langle \mathbf n, (\Hess f)(\mathbf x)^{-1} \mathbf n \rangle} \cdot
e^{-f(\mathbf x)}\,d m_{n-1}(\mathbf x)}
{\alpha(s)}
+ o(\gamma_1^2)
\end{multline*}
as $\gamma_1\searrow 0$. Thus, since $(\ln \alpha)''(s)$ is well defined,
\begin{equation*}
  -(\ln \alpha)''(s) \ge \frac
{\int\limits_{X(\gamma,\eps,\delta)}
\frac{1 - \varepsilon}
{2 \langle \mathbf n, (\Hess f)(\mathbf x)^{-1} \mathbf n \rangle} \cdot
e^{-f(\mathbf x)}\,d m_{n-1}(\mathbf x)}
{\alpha(s)}.
\end{equation*}
The lemma follows by taking $\delta\searrow 0$ followed by $\gamma\searrow 0$ and lastly $\eps\searrow 0$.
\end{proof}

The next lemma combined with part~\ref{item:logconcave_max} of Proposition~\ref{prop:1d_logconcave_properties} implies Theorem~\ref{thm:quantile Brascamp-Lieb}.

\begin{lem}\label{lem:app:main}
Let $X$ be a random vector with a log-concave density $\exp(-f)$. Let $\mathbf n$ be a unit vector in $\R^n$. Let $t>0$ and set
\begin{equation}\label{eq:app:appl_condition}
p:=\Pr \left( \langle \mathbf n, (\Hess f)(X)^{-1} \mathbf n \rangle
\leq t \right).
\end{equation}
Let $\alpha$ be the (log-concave) density of $\langle\mathbf n,X\rangle$. Then
\begin{equation}\label{eq:app:appl_conclusion}
\max\limits_{s \in \mathbb R} \alpha(s)
\geq \frac{p^{3/2}}{(8 - 4p)\sqrt{2t}}.
\end{equation}
\end{lem}

\begin{proof}
We first prove the statement under the additional assumption that $(\Hess f)(X)$ is positive definite almost surely. With this assumption we may apply Lemma~\ref{lem:app:second_derivative} to conclude that
\begin{equation}
  -(\ln \alpha)''(\langle\mathbf n, X\rangle)\ge \E\left(\frac{1}{\langle \mathbf n, (\Hess f)(X)^{-1} \mathbf n \rangle}\,|\, \langle\mathbf n, X\rangle\right)
\end{equation}
almost surely. In particular,
\begin{equation}\label{eq:lower bound with conditional probability}
  -(\ln \alpha)''(\langle\mathbf n, X\rangle)\ge \frac{1}{t}\cdot \P(\langle \mathbf n, (\Hess f)(X)^{-1} \mathbf n \rangle\le t\,|\, \langle\mathbf n, X\rangle).
\end{equation}
Markov's inequality shows that for any random variable $0\le \theta\le 1$ one has $\P\left(\theta>\frac{\E(\theta)}{2}\right)> \frac{\E(\theta)}{2-\E(\theta)}$. Applying this to the random variable $\P(\langle \mathbf n, (\Hess f)(X)^{-1} \mathbf n \rangle\le t\,|\, \langle\mathbf n, X\rangle)$, whose expectation is $p$ by~\eqref{eq:app:appl_condition}, we may continue~\eqref{eq:lower bound with conditional probability} to obtain
\begin{equation}\label{eq:lower bound second derivative probability}
  \P\left(-(\ln \alpha)''(\langle\mathbf n, X\rangle)> \frac{p}{2t}\right)> \frac{p}{2-p}.
\end{equation}
This is to be compared with the conclusion of Lemma~\ref{lem:app:max_density}, which states that for any $C\ge4$,
\begin{equation}\label{eq:upper bound second derivative probability}
  \P\left(-(\ln \alpha)''(\langle\mathbf n, X\rangle)> \left(C\max\limits_{s \in \mathbb R} \alpha(s)\right)^2\right)\le 4C^{-1}.
\end{equation}
Substituting $C = \frac{8 - 4p}{p}$, the two inequalities show that $\left(\frac{8 - 4p}{p}\max\limits_{s \in \mathbb R} \alpha(s)\right)^2> \frac{p}{2t}$, proving~\eqref{eq:app:appl_conclusion}, under our additional assumption on the positivity of $\Hess f$.

To treat the case of general $f$, define, for $\eps>0$,
\begin{equation*}
  f_\eps(\mathbf x):= f(\mathbf x) + \frac{1}{2}\eps\|\mathbf x\|^2 - \ln Z_{\eps}
\end{equation*}
where $Z_\eps=\int \exp\left(-f(\mathbf x) - \frac{1}{2}\eps\|\mathbf x\|^2\right)\,d\mathbf x$ is chosen so that $\int f_\eps(\mathbf x)\,d\mathbf x = 1$. Note that $Z_\eps\to 1$ as $\eps\searrow 0$ by the dominated convergence theorem. Let $X_\eps$ be a random vector with density $f_\eps$. Note that $\Hess f_{\eps}(\mathbf x) = \Hess f(\mathbf x) + \eps Id$ for almost every $\mathbf x$. In particular, $\langle \mathbf n, (\Hess f)(\mathbf x)^{-1} \mathbf n \rangle\ge \langle \mathbf n, (\Hess f_\eps)(\mathbf x)^{-1} \mathbf n \rangle$ almost everywhere. Thus
\begin{equation*}
\begin{split}
  p=\Pr \left( \langle \mathbf n, (\Hess f)(X)^{-1} \mathbf n \rangle
\leq t \right) &= \lim_{\eps\searrow0} \Pr \left( \langle \mathbf n, (\Hess f)(X_\eps)^{-1} \mathbf n \rangle
\leq t \right)\\
&\le\liminf_{\eps\searrow0} \Pr \left( \langle \mathbf n, (\Hess f_\eps)(X_\eps)^{-1} \mathbf n \rangle
\leq t \right)
\end{split}
\end{equation*}
by a second application of the dominated convergence theorem. Now set $\alpha_\eps$ to be the log-concave density of $\langle \mathbf n, X_\eps\rangle$. It is straightforward that for every $s\in\R$,
\begin{equation*}
  \alpha(s) \ge \frac{1}{Z_\eps}\alpha_\eps(s).
\end{equation*}
Combining the above facts we obtain the conclusion~\eqref{eq:app:appl_conclusion} as a consequence of the theorem applied to $X_\eps$, by taking the limit $\eps\searrow 0$.
\end{proof}

\begin{proof}[Proof of Theorem~\ref{thm:quantile Brascamp-Lieb}]
Define the unit vector $\mathbf n:=\frac{\eta}{\|\eta\|}$ and let $\alpha$ be the log-concave density of $\langle \mathbf n,X\rangle$. Let $r>0$ and set
\begin{equation}
p:=\Pr \left( \langle \mathbf n, (\Hess f)(X)^{-1} \mathbf n \rangle
\leq r \right).
\end{equation}
Lemma~\ref{lem:app:main} shows that
\begin{equation}
  \max\limits_{s \in \mathbb R} \alpha(s) \geq \frac{p^{3/2}}{(8 - 4p)\sqrt{2r}}\ge \frac{p^{3/2}}{8\sqrt{2r}}.
\end{equation}
Part~\ref{item:logconcave_max} of Proposition~\ref{prop:1d_logconcave_properties} then implies that
\begin{equation}
  \Var(\langle \mathbf n, X\rangle) \le \frac{C r}{p^3}
\end{equation}
for a universal constant $C>0$. The conclusion~\eqref{eq:quantile Brascamp Lieb conclusion} now follows by setting $r = \frac{t}{\|\eta\|^2}$.
\end{proof}

\subsection{Proof of the quantitative log-concavity theorem}
In this section we prove Theorem~\ref{thm:quantitative log concavity}.

As in the theorem, let $X$ be a random vector with a log-concave density $\exp(-f)$ and $\eta\in\R^n$. Denote by $\alpha_\eta:\R\to[0,\infty)$ the (log-concave) density function of $\langle \eta, X\rangle$. Fix $D\ge0$, $t>0$ and $s\in\R$ satisfying that $\alpha_\eta(s)>0$.

Denote by $m_{n-1}$ the $(n-1)$-dimensional Hausdorff measure on $\R^n$. We fix a representative of the marginal density $\alpha_\eta$ by setting
\begin{equation}
  \alpha_\eta(r) = \int\limits_{\{x\in\R^n\colon \langle \eta,x\rangle = r\}} \exp(-f(x))d m_{n-1}(x),\quad r\in\R.
\end{equation}
Similarly, we fix a representative of $\gamma_\eta(D,\cdot, t)$ by setting
\begin{equation}\label{eq:gamma eta representative}
  \gamma_\eta(D,r,t)\alpha_\eta(r) = \int\limits_{\{x\in\R^n\colon \langle \eta,x\rangle = r,\, D_{\eta,x}(t)\ge D\}} \exp(-f(x))d m_{n-1}(x)
\end{equation}
for $r$ in the open interval $\{r\in\R\colon \alpha_\eta(r)>0\}$.

Recall the definition of $D_{\eta,x}$ from~\eqref{eq:E_f_2_def} and note that it takes values in $[0,\infty]$ by the convexity of $f$. Note further that $D_{\eta,\cdot}(t)$ is upper semi-continuous and, in particular, measurable.

We first prove~\eqref{eq:alpha_f_2_bound}. Aiming to apply the Pr\'ekopa--Leindler inequality, define $F_1, F_2:\R^{n-1}\to[0,\infty)$ to be the restrictions of the density $\exp(-f(x))$ to the hyperplanes $\langle \eta,x\rangle = s+t$ and $\langle \eta,x\rangle = s-t$, respectively (the hyperplanes are parameterized as $\R^{n-1}$ and are equipped with a standard Lebesgue measure - the projection of $m_{n-1}$). Set $g:\R^n\to(-\infty,\infty]$ to equal $f(x) + D_{\eta,x}(t)$ for $x$ satisfying $f(x)<\infty$ and to equal $\infty$ for other $x$. Lastly, define $F:\R^{n-1}\to[0,\infty)$ as the restriction of $\exp(-g(x))$ to the hyperplane $\langle \eta,x\rangle = s$.

The above definitions imply that the assumption~\eqref{eq:Prekopa-Leindler_condition} of the Pr\'ekopa--Leindler inequality with $\lambda = \frac{1}{2}$ is satisfied, i.e.,
\begin{equation}
  F\left(\frac{1}{2}(\mathbf x + \mathbf y)\right) \ge \sqrt{F_1(\mathbf x)F_2(\mathbf y)},\quad \mathbf x, \mathbf y\in\R^{n-1}.
\end{equation}
Consequently,
\begin{equation}
  \int F(\mathbf x)d \mathbf x
    \ge \sqrt{\int F_1(\mathbf x)d \mathbf x \int F_2(\mathbf x)d \mathbf x} = \sqrt{\alpha_\eta(s+t)\alpha_\eta(s-t)}.
\end{equation}
To prove~\eqref{eq:alpha_f_2_bound}, it remains to note that
\begin{equation}
  \int F(\mathbf x)d \mathbf x\le \bigl(1-\gamma_{\eta}(D, s, t) (1- e^{-D}) \bigr) \cdot \alpha_{\eta}(s).
\end{equation}
Indeed, by the definition of $g$ and~\eqref{eq:gamma eta representative},
\begin{equation*}
  \int F(\mathbf x)d \mathbf x = \int\limits_{\{x\in\R^n\colon \langle \eta,x\rangle = s\}} \exp(-g(x))d m_{n-1}(x)\le ((1 - \gamma_\eta(D,s,t)) + \exp(-D)\gamma_\eta(D,s,t))\alpha_\eta(s).
\end{equation*}

We proceed to prove~\eqref{eq:Markov consequence}. Define the random variable $\Gamma:=\P(D_{\eta,X}(t)\ge D\,|\,\langle \eta, X\rangle)$. Then, by~\eqref{eq:alpha_f_2_bound},
\begin{equation}\label{eq:alpha_f_2_bound_random}
\sqrt{\alpha_{\eta}(\langle \eta, X\rangle-t)\alpha_{\eta}(\langle \eta, X\rangle+t)} \le \bigl(1-\Gamma\cdot (1- e^{-D}) \bigr) \cdot \alpha_{\eta}(\langle \eta, X\rangle).
\end{equation}
Hence it suffices to prove that $\P(\Gamma\ge\frac{1}{2})\ge\frac{1}{2}$. This follows from the fact that $\E(\Gamma)\ge \frac{3}{4}$ (by Markov's inequality applied to $1-\Gamma$).

\section{Random surface preliminaries}\label{sec:random surface preliminaries}
In this section we develop the technical tools necessary to prove our main results on random surfaces, Theorem~\ref{thm:main} and Theorem~\ref{thm:tail_estimate}. These tools will be put together in the following Section~\ref{sec:proofs of results on random surfaces} to enable the use of the concentration results of Section~\ref{sec:concentration inequalities} in the random surface context.

\subsection{Isoperimetric properties of $\Lambda_L^d$}\label{sec:isoperimetry}

In this section we provide isoperimetric estimates for two graphs: the box $\Lambda_L^d$ and the graph obtained by performing a bond percolation process on $\Lambda_L^d$. The bond percolation process we will need for our application to random surfaces is not the usual independent percolation, but is still super-critical in a suitable sense.
We postpone the formal definition; it is given later
by~\eqref{eq:subset_prob}.

For a graph $G = (V(G), E(G))$ and a subset $X \subset V(G)$ write
\begin{align*}
  &E(G) \vert_X := \{ \{ x, y \} \in E(G) : \{ x, y \} \subseteq X \} \\
  &\partial_G X := \{ \{ x, y \} \in E(G) : \# (\{ x, y \} \cap X) = 1 \}.
\end{align*}
If $X \neq \varnothing$, let $G \vert_X := (X, E(G) \vert_X)$ be the induced subgraph of $G$ on $X$. For connected $G$, we will further use the class of connected subsets whose complement is also connected,
\begin{equation*}
  \mathcal C(G) := \{ X \subseteq V(G) :
  \text{$X \notin\{\varnothing, V(G)\}$,
  both $G \vert_X$ and $G \vert_{V(G) \setminus X}$ are connected} \}.
\end{equation*}

We write $|S|$ for the cardinality of a finite set $S$.

\subsubsection{The box $\Lambda_L^d$} The following is the isoperimetric inequality we need.
\begin{lem}\label{lem:isoperimetry of cube}
Let $d\ge 2$ and $L\ge 1$. If $X \subseteq V(\Lambda_L^d)$
and $\left| X \right| \le 3L^d / 4$, then
\begin{equation*}
  \left| \partial_{\Lambda_L^d} X \right| \ge
  \left |X \right|^{\frac{d-1}{d}}.
\end{equation*}
\end{lem}
The inequality follows as a corollary of the following result.

\begin{lem}\label{lem:bollobas-leader-ineq}(Bollob\'as and Leader~\cite{BL91})
Let $d\ge 2$ and $L\ge 1$. Given $X \subseteq V(\Lambda_L^d)$, it holds
that
\begin{align*}
(i) & \quad \left| \partial_{\Lambda_L^d} X \right| \ge
\min\limits_{r \in \{ 1, 2, \ldots, d \} }
\left| X \right|^{1 - \frac{1}{r}}r L^{\frac{d}{r}-1} & \quad \text
{if $\left| X \right| \leq L^d / 2$}; \\
(ii) & \quad \left| \partial_{\Lambda_L^d} X \right| \ge L^{d - 1}
& \quad \text{if $L^d / 4 \leq \left| X \right| \leq 3L^d / 4$}.
\end{align*}
\end{lem}
These results appear in~\cite{BL91} as Theorem~3 (item (i)) and Corollary~4 (item (ii)).
\begin{proof}[Proof of Lemma~\ref{lem:isoperimetry of cube}]
If $\left| X \right| \le L^d / 2$, then assertion (i) of
Lemma~\ref{lem:bollobas-leader-ineq} implies
\begin{equation*}
\left| \partial_{\Lambda_L^d} X \right|
\ge \frac{|X|}{L} \min\limits_{r \in \{ 1, 2, \ldots, d \} }
r \left(\frac{L^d}{|X|}\right)^{\frac{1}{r}}
\ge \frac{|X|}{L} \min\limits_{r \in \{ 1, 2, \ldots, d \} }
\left(\frac{L^d}{|X|}\right)^{\frac{1}{r}} = |X|^{\frac{d-1}{d}}.
\end{equation*}
In the remaining case, $L^d / 2 < |X| \le 3L^d / 4$, assertion (ii)
of Lemma~\ref{lem:bollobas-leader-ineq} yields
\begin{equation*}
|\partial_{\Lambda_L^d} X|
\ge L^{d - 1} > |X|^{\frac{d-1}{d}}. \qedhere
\end{equation*}
\end{proof}

\subsubsection{Bond percolation on $\Lambda_L^d$}
In this section we study isoperimetry for random spanning subgraphs of $\Lambda_L^d$ (i.e., random subgraphs whose vertex set is $V(\Lambda_L^d)$). Specifically, for $0<p<1$ we consider random spanning subgraphs $\Lambda_{L,p}^d$ of $\Lambda_L^d$ satisfying the property
\begin{equation}\label{eq:subset_prob}
  \Pr \left(E'\cap E(\Lambda_{L,p}^d) = \varnothing \right)
  \le (1-p)^{|E'|}\quad\text{for each $E'\subset E(\Lambda_L^d)$}.
\end{equation}
Independent bond percolation with parameter $p$ certainly satisfies~\eqref{eq:subset_prob} but our application to random surfaces will make use of a more general percolation process, for which~\eqref{eq:subset_prob} still holds true. The next lemma studies the (anchored) isoperimetric properties of $\Lambda_{L,p}^d$ for $p$ sufficiently close to $1$ (see also~\cite{Pe08} for related statements).

\begin{lem}\label{lem:isoperimetry of cube after percolation}
Let $d\ge 2$. There is a function $q:(0,1)\to [0,1]$ satisfying
$\lim_{p \nearrow 1} q(p)=1$ such that the following holds. Let $L\ge 1$
and let $a,b \in V(\Lambda_L^d)$. Let $0<p<1$. Suppose
$\Lambda_{L,p}^d$ is a random spanning subgraph of $\Lambda_L^d$ satisfying~\eqref{eq:subset_prob}.
Let $G_a$ denote the connected component of $a$ in $\Lambda_{L,p}^d$
(considered as a connected graph). Then each of the following events holds with probability at least $q(p)$:
\begin{enumerate}
  \item \label{item:isoperimetry1}$b \in V(G_a)$.
  \item \label{item:isoperimetry2}For each set $X\in\mathcal C(G_a)$ satisfying $a\in X$ and
  $|X| \le \frac{1}{2} \left| V(G_a) \right|$ it holds that
\begin{equation*}
  \left| \partial_{G_a} X \right| \ge \frac{1}{2}|X|^{\frac{d-1}{d}}.
\end{equation*}
\end{enumerate}
\end{lem}

The proof of the lemma makes use of the following standard estimate for the number of connected sets with connected complement. For $a\in V(\Lambda_L^d)$ and integer $m\ge 1$ define
\begin{equation*}
  \mathcal C_{a,m}(\Lambda_L^d):=
  \left\{ X\in\mathcal C(\Lambda_L^d)\colon a\in X,\; |X| \le 3L^d / 4,\;
  \left| \partial_{\Lambda_L^d} X \right| = m \right\}.
\end{equation*}
\begin{lem}\label{lem:number_of_connected_sets}
Let $d\ge 2$ and $L\ge 1$. There exists $C>1$, depending only on $d$, such that the inequality
\begin{equation*}
  \left| \mathcal C_{a,m}(\Lambda_L^d) \right| \le C^m
\end{equation*}
holds for each $a\in V(\Lambda_L^d)$ and every integer $m\ge 1$.
\end{lem}
As we could not find a reference for this exact statement, we provide a proof at the end of the section.
\begin{proof}[Proof of Lemma~\ref{lem:isoperimetry of cube after percolation}]
For each integer $m\ge 1$ and $v\in \Lambda_L^d$, define the event
\begin{equation*}
    \mathcal E_{v,m} := \left\{\exists X \in \mathcal C_{v,m}(\Lambda_L^d)\text{ satisfying }\left| \partial_{\Lambda_{L,p}^d} X \right| < \tfrac{1}{2}m \right\}.
\end{equation*}
Define also $\mathcal E_v:= \bigcup\limits_{m = 1}^{\infty}
\mathcal E_{v,m}$. We first prove the convergence
\begin{equation}\label{eq:E_v_prob_bound}
  \lim_{p \to 1} \Pr(\mathcal E_v) = 0 \quad\text{uniformly in $L,v$},
\end{equation}
where here and later we also implicitly require the uniformity of the
statements in the choice of processes $(\Lambda_{L,p}^d)$
satisfying~\eqref{eq:subset_prob}. To see this, first consider a specific
$X\in \mathcal C_{v,m}$. The inequality
$\left| \partial_{\Lambda_{L,p}^d} X \right| < \tfrac{1}{2}m$ implies that some $\lceil m/2\rceil$-edge subset of $\partial_{\Lambda_{L}^d} X$
is completely removed when passing from $\Lambda^d_L$ to
$\partial_{\Lambda_{L,p}^d} X$. Taking a union bound over such subsets,
we obtain
\begin{equation*}
  \Pr\left(\left| \partial_{\Lambda_{L,p}^d} X \right| <
  \tfrac{1}{2}m \right)
  \le \binom{m}{\lceil m/2\rceil} (1-p)^{\lceil m/2\rceil}
  \le (4(1-p))^{m/2}.
\end{equation*}
Now, if $C$ is the constant from Lemma~\ref{lem:number_of_connected_sets},
and if $p$ is sufficiently close to $1$, then ~\eqref{eq:E_v_prob_bound}
is implied by the following inequalities:
\begin{equation*}
  \Pr(\mathcal E_v)
  \le \sum_{m=1}^\infty \Pr(\mathcal E_{v,m})
  \le \sum_{m=1}^\infty (4C^2(1-p))^{m/2} \le 3C\sqrt{1-p}.
\end{equation*}

We next prove that
\begin{equation}\label{eq:v_not_in_G_a_prob}
  \lim_{p\to 1}\Pr(v\notin V(G_a)) = 0
  \quad\text{uniformly in $L,a$ and $v\in V(\Lambda_L^d)$}.
\end{equation}
To prove~\eqref{eq:v_not_in_G_a_prob}, we assume that for some
$v\in V(\Lambda_L^d)$ the property $v\notin V(G_a)$ holds.
Then there exists a set $X\in \mathcal C(\Lambda_L^d)$ with $a\in X$,
$v\notin X$ and $\partial_{\Lambda_{L,p}^d} X=\varnothing$.
Further, $\min \left\{ |X|, \left| V(\Lambda_L^d) \setminus X \right|
\right\} \leq L^d / 2 \leq 3L^d / 4$, whence the event
$\mathcal E_a \cup \mathcal E_v$ holds.
(Indeed, $\mathcal E_a$ holds if $|X| \leq 3L^d / 4$
and $\mathcal E_v$ holds if
$\left| V(\Lambda_L^d) \setminus X \right| \leq 3L^d / 4$). Therefore
$\Pr(v\notin V(G_a)) \leq \Pr(\mathcal E_a) + \Pr(\mathcal E_v)$,
and~\eqref{eq:v_not_in_G_a_prob} follows from~\eqref{eq:E_v_prob_bound}.
Taking $v = b$ proves the assertion of the lemma regarding item~\eqref{item:isoperimetry1}.

In order to prove the assertion of the lemma regarding item~\eqref{item:isoperimetry2},
define the event
\begin{equation*}
  \mathcal E := \left\{ \left| V(G_a) \right| < L^d / 2 \right\}.
\end{equation*}
We note that
\begin{equation}\label{eq:large_component_prob}
  \lim_{p\to 1} \Pr(\mathcal E) = 0\quad\text{uniformly in $L,a$}
\end{equation}
as a consequence of~\eqref{eq:v_not_in_G_a_prob} and the following
chain of inequalities:
\begin{equation*}
  \Pr(\mathcal E) \leq \frac{2}{L^d}
  \cdot \E \left| V(\Lambda_L^d)\setminus V(G_a) \right|
  = \frac{2}{L^d} \sum\limits_{v \in V(\Lambda_L^d)} \Pr(v \notin V(G_a))
  \leq 2 \cdot \max\limits_{v \in V(\Lambda_L^d)} \Pr(v \notin V(G_a)).
\end{equation*}
It thus suffices to show that the event in item~\eqref{item:isoperimetry2} holds when neither
$\mathcal E_a$ nor $\mathcal E$ hold. Let $X\in\mathcal C(G_a)$ for
which $a\in X$ and $|X| \le \frac{1}{2} \left| V(G_a) \right|$.
Let us extend $X$ to a set $X' = X \sqcup Y$, where
\begin{equation*}
Y := \{ w \in V(\Lambda_L^d) \setminus V(G_a): \text{$\nexists$ path in
$\Lambda_L^d \vert_{V(\Lambda_L^d) \setminus X}$ between $w$ and
$V(G_a) \setminus X$} \}
\end{equation*}
(informally, $X'$ is obtained from $X$ by ``filling holes''). From
the definition of $Y$ it is evident that
$X' \in \mathcal C(\Lambda_L^d)$, $V(G_a) \cap X' = X$ and
$\partial_{\Lambda_{L,p}^d} X' = \partial_{G_a} X$. Additionally,
when $\mathcal E$ does not hold, we have
\begin{equation*}
  |X'| \le L^d - \left| V(G_a) \setminus X \right|
  \leq L^d - \tfrac{1}{2} \left| V(G_a) \right| \le 3L^d / 4.
\end{equation*}
Lastly, using Lemma~\ref{lem:isoperimetry of cube} and the fact that
$\mathcal E_a$ does not hold, we conclude that
\begin{equation*}
  \left| \partial_{G_a} X \right| =
  \left| \partial_{\Lambda_{L,p}^d} X' \right| \ge
  \frac{1}{2}\left| \partial_{\Lambda_{L}^d} X' \right| \ge
  \frac{1}{2}|X'|^{\frac{d-1}{d}} \ge
  \frac{1}{2}|X|^{\frac{d-1}{d}},
\end{equation*}
so the event in item~\eqref{item:isoperimetry2} indeed occurs, provided that neither $\mathcal E$
nor $\mathcal E_a$ holds.
\end{proof}

We now give the proof of Lemma~\ref{lem:number_of_connected_sets}.

If $X \in \mathcal C(\Lambda^d_L)$, then, in a certain sense,
the boundary of $X$ is connected. This is established in the next lemma which is proved by Deuschel and Pisztora~\cite[part (ii) of Lemma 2.1]{DP96} (see also the related paper of Tim\'ar~\cite{T13}).
\begin{lem}\label{lem:deuschel-pisztora-connectivity}
Let $\Pi_L^d$ be the graph defined by
\begin{align*}
V(\Pi_L^d) & := E(\Lambda_L^d), \\
\{ e_1, e_2 \} \in E(\Pi_L^d) &\text{ if the midpoints of $e_1$ and $e_2$ are at $\ell^\infty$ distance in $(0,1]$}.
\end{align*}
Let $X \in \mathcal C(\Lambda^d_L)$. Then the induced subgraph
$\Pi_L^d \vert_{\partial X}$ is connected, where $\partial X$ stands
for $\partial_{\Lambda^d_L} X$.
\end{lem}

The following standard lemma gives a bound on the number of connected subsets of a graph containing a given vertex. A proof may be found in~\cite[Chapter 45]{B06}.
\begin{lem}\label{lem:enumeration of connected subsets}
  Let $G = (V(G), E(G))$ be a graph with maximal degree $\Delta\ge 3$. Let $v\in V(G)$ and integer $m\ge 1$. Then
  \begin{equation*}
    |\{S \subset V(G)\colon \text{$S$ is connected, $v\in S$ and $|S|=m$}\}|\le (e(\Delta-1))^{m-1}.
  \end{equation*}
\end{lem}

\begin{proof}[Proof of Lemma~\ref{lem:number_of_connected_sets}]
Since $X \in \mathcal C_{a, m}(\Lambda_L^d)$ is uniquely determined by
$\partial X := \partial_{\Lambda_L^d} X$, it is sufficient to enumerate
all possibilities for $\partial X$. Given an edge $e_1\in\partial X$, the connectivity property of Lemma~\ref{lem:deuschel-pisztora-connectivity} and the counting estimate of Lemma~\ref{lem:enumeration of connected subsets} show that there are at most $C^m$ options for $\partial X$, where $C$ depends only on $d$. To find the starting edge $e_1$ the following argument may be used. Choose $X_0 \subseteq V(\Lambda^d_L)$ so
that $a \in X_0$, $\Lambda^d_L \vert_{X_0}$ is connected and
$\left| X_0 \right| = \left\lfloor \min \left\{ m^{\frac{d}{d - 1}},
3L^d / 4\right\} \right\rfloor + 1$. Lemma~\ref{lem:isoperimetry of cube} implies that $|X| < \left| X_0 \right|$,
from which we conclude that $\partial X \cap E(\Lambda^d_L \vert_{X_0}) \neq \varnothing$.
Consequently, $e_1$ may be chosen as one of at most $|E(\Lambda^d_L \vert_{X_0})|\le d \left(m^{\frac{d}{d - 1}} + 1\right)\le \tilde{C}^m$ options for a constant $\tilde{C}$ depending only on $d$.
\end{proof}

\subsection{Energy estimate via isoperimetry}\label{sec:energy}

The following lemma is the main result of this section. It is a close
relative of the result by Benjamini and Kozma~\cite[Theorem 2.1]{BK05},
with the difference that our result is applicable for non-quadratic
potentials.

\begin{lem}\label{lem:energy_estimate_2}
Let $G$ be a finite connected graph. Let $l$ be an integer such that
$2^l \leq |V(G)| < 2^{l + 1}$. For each
$i \in \{1, 2, \ldots, l \}$ and for each vertex $v \in V(G)$ write
\begin{equation}\label{eq:m_isoperimetry}
\begin{aligned}
M_i(v) & :=  \max \left\{ |E(G \vert_X)| + |\partial_G X|
\; : \;
X \in \mathcal C(G), \; v \in X, \;
\left\lfloor \frac{|V(G)|}{2^{i + 1}} \right\rfloor < |X|
\leq \left\lfloor \frac{|V(G)|}{2^i} \right\rfloor \right\}, \\
m_i(v) & :=  \min \left\{ |\partial_G X|
\; : \;
X \in \mathcal C(G), \; v \in X, \;
\left\lfloor \frac{|V(G)|}{2^{i + 1}} \right\rfloor < |X|
\leq \left\lfloor \frac{|V(G)|}{2^i} \right\rfloor \right\},
\end{aligned}
\end{equation}
with the convention that a maximum, or minimum, over an empty collection remains undefined.
Let, finally, $U : \mathbb R \to \mathbb R$ be a convex function with
$U(0) = 0$ and $U(-x) \equiv U(x)$. Then the inequality
\begin{multline}\label{eq:energy_ineq_2}
\inf\limits_{ \substack{\varphi : V(G) \to \mathbb R \\
\varphi(a) - \varphi(b) = 1} } \;
\sum \limits_{e \in E(G)} U(\nabla_e \varphi) \geq \\
\inf\limits_{ \substack{p_1, \ldots, p_l \geq 0 \\
q_1, \ldots, q_l \geq 0 \\
p_1 + \ldots + p_l + q_1 + \ldots + q_l = 1}}
\left( \sum\limits_{i = 1}^{l}
M_i(a) \cdot U \left(\frac{p_i m_i(a)}{M_i(a)} \right) +
\sum\limits_{i = 1}^{l}
M_i(b) \cdot U \left(\frac{q_i m_i(b)}{M_i(b)} \right) \right)
\end{multline}
holds for every two distinct vertices $a, b \in V(G)$, when all quantities $M_i(a), m_i(a),M_i(b),m_i(b)$ are defined. In the presence of undefined terms the inequality continues to hold with the following modification: Whenever one of $M_i(a), m_i(a)$ (respectively, $M_i(b), m_i(b)$) is undefined the corresponding summand is set to 0 and the additional restriction $p_i=0$ (respectively, $q_i = 0$) is added to the infimum.
\end{lem}

We remark that while some of the terms $M_i(v), m_i(v)$ may indeed be undefined for some graphs $G$ and vertices $v\in V(G)$, it is always the case that for every distinct $a,b\in V(G)$ there is an $X \in \mathcal C(G)$ with $a\in X$ and $b\notin X$, implying that for at least one $i$ either both $M_i(a), m_i(a)$ or both $M_i(b), m_i(b)$ are defined.

Before we proceed with the proof, let us show that a well-known energy
bound for the quadratic potential on graphs with the isoperimetry of a
$\Z^d$ lattice follows as an immediate corollary.

\begin{cor}\label{cor:energy_estimate_quadratic}
In the setup of Lemma~\ref{lem:energy_estimate_2}, suppose that the graph
$G$ and vertices $a, b \in V(G)$ are such that the quantities $M_i(a), m_i(a), M_i(b), m_i(b)$ which are defined satisfy the inequalities
\begin{equation*}
M_{l-i}(v) \le C\, 2^i \quad \text{and} \quad
m_{l-i}(v)\ge c\, 2^{\frac{d-1}{d} i}\quad
\text{for $1\le i\le l$ and $v\in\{a,b\}$},
\end{equation*}
where $d\ge 2$ is an integer and $C,c>0$. Then there exists
$c' = c'(d,C,c)>0$ such that
\begin{equation*}
\inf\limits_{ \substack{\varphi : V(G) \to \mathbb R \\
\varphi(a) - \varphi(b) = 1} } \;
\sum \limits_{e \in E(G)} \left(\nabla_e \varphi\right)^2 \geq \left\{
\begin{array}{ll}
\frac{c'}{l} & \quad \text {if $d = 2$,}\\
c'& \quad \text{if $d \ge 3$.}
\end{array} \right.
\end{equation*}
\end{cor}
\begin{proof}
Lemma~\ref{lem:energy_estimate_2} implies that for some
$c_0 = c_0(d, C, c)$ we have
\begin{equation*}
\inf\limits_{ \substack{\varphi : V(G) \to \mathbb R \\
\varphi(a) - \varphi(b) = 1} } \;
\sum \limits_{e \in E(G)} \left(\nabla_e \varphi\right)^2
\ge c_0 \inf\limits_{\substack{p_1, \ldots, p_l \geq 0 \\
q_1, \ldots, q_l \geq 0 \\
p_1 + \ldots + p_l + q_1 + \ldots + q_l = 1}}
\left( \sum\limits_{i = 0}^{l-1} 2^{i(1-2/d)}
(p_{l-i}^2 + q_{l-i}^2) \right).
\end{equation*}
The corollary follows by using the Cauchy--Schwarz inequality in the form $(\sum r_i)^2 \le (\sum a_i r_i^2)(\sum \frac{1}{a_i})$, where
$r_{2i - 1} = p_{l - i}$, $r_{2i} = q_{l - i}$,
$a_{2i - 1} = a_{2i} = 2^{i(1 - 2/d)}$ for $i = 1, 2, \ldots, l$.
\end{proof}

\begin{proof}[Proof of Lemma~\ref{lem:energy_estimate_2}] The proof is given for the case that all the quantities $M_i(a), m_i(a),M_i(b),m_i(b)$ are defined. The required modifications when some of the terms are undefined are straightforward.

Note that $U$ has to be continuous on $\mathbb R$, since it is
convex and attains only finite values.

No generality is lost if we assume $\varphi(a) = 1$ and $\varphi(b) = 0$.
But then the inequality
\begin{equation*}
\sum \limits_{e \in E(G)} U(\nabla_e \varphi) \geq
\sum \limits_{e \in E(G)} U(\nabla_e \bar{\varphi})
\end{equation*}
holds, where
\begin{equation*}
\bar{\varphi}(v) := \left\{
\begin{array}{ll}
0, & \quad \text{if $\varphi(v) < 0$} \\
\varphi(v), & \quad \text{if $0 \leq \varphi(v) \leq 1$} \\
1, & \quad \text{if $\varphi(v) > 1$}
\end{array}
\right. \qquad \text{for all $v \in V(G)$}.
\end{equation*}
Thus it is sufficient to consider $\varphi \in [0, 1]^{V(G)}$. For the
reasons of compactness, the infimum
\begin{equation*}
\inf\limits_{ \substack{\varphi : V(G) \to \mathbb R \\
\varphi(a) - \varphi(b) = 1} } \;
\sum \limits_{e \in E(G)} U(\nabla_e \varphi)
\end{equation*}
is attained on a non-empty compact set
$\Omega_{\min} \subseteq [0, 1]^{V(G)}$. Let
\begin{equation*}
\varphi_{\min} = \argmin\limits_{\varphi \in \Omega_{\min}}
\sum \limits_{e \in E(G)} (\nabla_e \varphi)^2
\end{equation*}
(the quadratic function is used here for convenience and can be replaced by other strictly convex, even functions on $\R$).
Then for every $s \in [0, 1)$ it holds that $X(s) \in \mathcal C(G)$,
where
\begin{equation*}
X(s) := \{ v \in V(G) : \varphi_{\min}(v) \leq s \}.
\end{equation*}
Indeed, assume for the contradiction that $G \vert_{X(s)}$ is disconnected.
If $X_0 \subsetneq X(s)$ is such that $b \notin X_0$ and
$G \vert_{X_0}$ is a connected component of $G \vert_{X(s)}$, define
\begin{equation*}
\varphi_{\varepsilon}(v) := \left\{
\begin{array}{ll}
\varphi_{\min}(v), & \quad \text{if $v \notin X_0$} \\
\varphi_{\min}(v) + \varepsilon, & \quad \text{if $v \in X_0$}
\end{array}
\right. \qquad \text{for all $v \in V(G)$}.
\end{equation*}
Then, for all sufficiently small $\varepsilon > 0$,
\begin{equation*}
\sum \limits_{e \in E(G)} U(\nabla_e \varphi_{\varepsilon})
\leq \sum \limits_{e \in E(G)} U(\nabla_e \varphi_{\min})
\quad \text{and} \quad
\sum \limits_{e \in E(G)} (\nabla_e \varphi_{\varepsilon})^2
< \sum \limits_{e \in E(G)} (\nabla_e \varphi_{\min})^2,
\end{equation*}
a contradiction to the choice of $\varphi_{\min}$. The induced graph
$G \vert_{V(G) \setminus X(s)}$ is connected for similar reasons.
From now on we will write $\varphi$ instead of $\varphi_{\min}$,
since this will not cause any confusion.

Given a number $\nu \in \mathbb N$ and an edge $e = \{ u, v \} \in E(G)$,
denote
\begin{equation*}
  \mathcal T_{\nu}(e) = \{0,1,2,\ldots, \nu-1\}\cap[\min \nu\cdot \varphi(e),\, \max \nu\cdot \varphi(e)),
\end{equation*}
with $\nu\cdot\varphi(e) := \{\nu \cdot\varphi(u), \nu\cdot \varphi(v)\}$, so that
\begin{equation*}
  \tau\in \mathcal T_{\nu}(e)\quad\text{if and only if}\quad e\in\partial_G\, X \left(\tau / \nu \right).
\end{equation*}
Then
\begin{equation*}
\sum \limits_{e \in E(G)} U(\nabla_e \varphi) =
\lim\limits_{\nu \to \infty}
\sum \limits_{e \in E(G)}
U\left( \frac{|\mathcal T_{\nu}(e)|}{\nu} \right).
\end{equation*}

Let us define a partition
\begin{equation*}
\{ 0, 1, 2, \ldots, \nu - 1 \} =
\mathcal I_{1, \nu}^a \sqcup \ldots \sqcup \mathcal I_{l, \nu}^a \sqcup
\mathcal I_{1, \nu}^b \sqcup \ldots \sqcup \mathcal I_{l, \nu}^b
\end{equation*}
according to the following equivalence relations:
\begin{equation}\label{eq:partition}
\arraycolsep=1.4pt\def\arraystretch{2.0}
\begin{array}{lll}
\tau \in \mathcal I_{i, \nu}^b \quad \Longleftrightarrow & \quad
\left\lfloor \frac{|V(G)|}{2^{i + 1}} \right\rfloor
< \left| X\left( \frac{\tau}{\nu} \right) \right|
\leq \left\lfloor \frac{|V(G)|}{2^i} \right\rfloor, &
\quad i = 1, 2, \ldots, l, \\
\tau \in \mathcal I_{1, \nu}^a \quad \Longleftrightarrow & \quad
\left\lfloor \frac{|V(G)|}{4} \right\rfloor
< \left| V(G) \setminus X\left( \frac{\tau}{\nu} \right) \right|
< \left\lceil \frac{|V(G)|}{2} \right\rceil, & \\
\tau \in \mathcal I_{i, \nu}^a \quad \Longleftrightarrow & \quad
\left\lfloor \frac{|V(G)|}{2^{i + 1}} \right\rfloor
< \left| V(G) \setminus X\left( \frac{\tau}{\nu} \right) \right|
\leq \left\lfloor \frac{|V(G)|}{2^i} \right\rfloor, &
\quad i = 2, 3, \ldots, l.
\end{array}
\end{equation}
It is not hard to check that indeed, according to the
definition~\eqref{eq:partition}, each number $\tau$ becomes assigned to
exactly one set $\mathcal I^a_{i, \nu}$ or $\mathcal I^b_{i, \nu}$.

Define
\begin{align*}
\mathcal P_{i, \nu}^a & := \bigl\{ (e, \tau) : e \in E(G), \;
\tau \in \mathcal I_{i, \nu}^a, \; \tau \in \mathcal T_{\nu}(e) \bigr\}, \\
\mathcal P_{i, \nu}^b & := \bigl\{ (e, \tau) : e \in E(G), \;
\tau \in \mathcal I_{i, \nu}^b, \; \tau \in \mathcal T_{\nu}(e) \bigr\}.
\end{align*}

For the next steps of the proof it will be useful to recall that since $U$ is convex and satisfies $U(0)=0$ then $U(t x)\le t\, U(x)$ for $t\in[0,1]$ and $x\ge 0$, and consequently $U(x_1 + \ldots + x_n) \geq U(x_1) + \ldots + U(x_n)$ for non-negative $x_1,\ldots, x_n$. Thus
\begin{multline*}
U\left( \frac{|\mathcal T_{\nu}(e)|}{\nu} \right) =
U \left( \frac{1}{\nu} \sum\limits_{i = 1}^l
\sum\limits_{\tau = 0}^{\nu - 1}
\mathbbm 1[(e, \tau) \in \mathcal P_{i, \nu}^a] +
\frac{1}{\nu} \sum\limits_{i = 1}^l \sum\limits_{\tau = 0}^{\nu - 1}
\mathbbm 1[(e, \tau) \in \mathcal P_{i, \nu}^b] \right) \\
\geq \sum\limits_{i = 1}^l
U\left( \frac{1}{\nu} \sum\limits_{\tau = 0}^{\nu - 1}
\mathbbm 1[(e, \tau) \in \mathcal P_{i, \nu}^a] \right) +
\sum\limits_{i = 1}^l
U\left( \frac{1}{\nu} \sum\limits_{\tau = 0}^{\nu - 1}
\mathbbm 1[(e, \tau) \in \mathcal P_{i, \nu}^b] \right).
\end{multline*}
Assume that $\mathcal I^b_{i, \nu} \neq \varnothing$. Set
$\tau^b_{i, \nu} := \max \mathcal I^b_{i, \nu}$ and $E^b_{i, \nu} :=
E \left( G \vert_{X\left( \tau^b_{i, \nu}/ \nu \right)} \right) \cup
\partial_G \Bigl( X\left( \tau^b_{i, \nu}/ \nu \right) \Bigr)$. By definition, $(e, \tau) \in \mathcal P_{i, \nu}^b$ if and only if $\tau \in \mathcal I^b_{i, \nu}$ and $e\in\partial_G\, X \left(\tau / \nu \right)$. Therefore, if $e\in E(G)$ is such that there exists $\tau$ with $(e, \tau) \in \mathcal P_{i, \nu}^b$ then $e\in E^b_{i, \nu}$. We conclude that
\begin{multline*}
\sum\limits_{e \in E(G)}
U\left( \frac{1}{\nu} \sum\limits_{\tau = 0}^{\nu - 1}
\mathbbm 1[(e, \tau) \in \mathcal P_{i, \nu}^b] \right) =
\sum\limits_{e \in E^b_{i, \nu}}
U\left( \frac{1}{\nu} \sum\limits_{\tau \in \mathcal I_{i, \nu}^b}
\mathbbm 1[(e, \tau) \in \mathcal P_{i, \nu}^b] \right) \geq \\
|E^b_{i, \nu}| \cdot
U \left( \frac{| \mathcal I_{i, \nu}^b |}{\nu} \cdot m_i(b) \cdot
\frac{1}{| E^b_{i, \nu} |} \right) \geq
M_i(b) \cdot U \left( \frac{| \mathcal I^b_{i, \nu} |}{\nu} \cdot m_i(b)
\cdot \frac{1}{M_i(b)} \right),
\end{multline*}
where the inequalities use Jensen's inequality and the inequalities
\begin{equation*}
| E^b_{i, \nu} | \leq M_i(b) \quad \text{and} \quad
\sum\limits_{e \in E^b_{i, \nu}}
\mathbbm 1[(e, \tau) \in \mathcal P_{i, \nu}^b]
= | \partial_G\, X \left(\tau / \nu \right) |\geq m_i(b) \quad
\text{for every $\tau \in \mathcal I_{i, \nu}^b$}.
\end{equation*}
If $\mathcal I^b_{i, \nu} = \varnothing$, then we still have
\begin{equation*}
\sum\limits_{e \in E}
U\left( \frac{1}{\nu} \sum\limits_{\tau = 0}^{\nu - 1}
\mathbbm 1[(e, \tau) \in \mathcal P_{i, \nu}^b] \right)  \geq
M_i(b) \cdot U \left( \frac{| \mathcal I^b_{i, \nu} |}{\nu} \cdot m_i(b)
\cdot \frac{1}{M_i(b)} \right),
\end{equation*}
since both sides of the inequality are zero.

Similarly, for each $i = 1, 2, \ldots, l$ it holds that
\begin{equation*}
\sum\limits_{e \in E}
U\left( \frac{1}{\nu} \sum\limits_{\tau = 0}^{\nu - 1}
\mathbbm 1[(e, \tau) \in \mathcal P_{i, \nu}^a] \right)  \geq
M_i(a) \cdot U \left( \frac{| \mathcal I^a_{i, \nu} |}{\nu} \cdot m_i(a)
\cdot \frac{1}{M_i(a)} \right),
\end{equation*}
since the roles of $a$ and $b$ are interchangeable.

The argument above yields
\begin{multline*}
\sum \limits_{e \in E(G)}
U\left( \frac{|\mathcal T_{\nu}(e)|}{\nu} \right) \\
\geq \sum\limits_{i = 1}^{l} M_i(a) \cdot
U \left(\frac{| \mathcal I_{i, \nu}^a |}{\nu} \cdot m_i(a) \cdot
\frac{1}{M_i(a)} \right) +
\sum\limits_{i = 1}^{l} M_i(b) \cdot
U \left(\frac{| \mathcal I_{i, \nu}^b |}{\nu} \cdot m_i(b) \cdot
\frac{1}{M_i(b)} \right) \\
\geq \inf\limits_{\substack{p_1, \ldots, p_l \geq 0 \\
q_1, \ldots, q_l \geq 0 \\
p_1 + \ldots + p_l + q_1 + \ldots + q_l = 1}}
\left( \sum\limits_{i = 1}^{l} M_i(a) \cdot
U \left(\frac{p_i m_i(a)}{M_i(a)} \right) +
\sum\limits_{i = 1}^{l} M_i(b) \cdot
U \left(\frac{q_i m_i(b)}{M_i(b)} \right) \right),
\end{multline*}
as the second inequality is implied by the identity $| \mathcal I_{1, \nu}^a | + \ldots + | \mathcal I_{l, \nu}^a | + | \mathcal I_{1, \nu}^b | + \ldots + | \mathcal I_{l, \nu}^b | = \nu$.
The inequality~\eqref{eq:energy_ineq_2} follows by passing to
the limit $\nu \to \infty$.
\end{proof}

\subsection{Estimates on the joint distribution of gradients}\label{sec:sparsity}

In this section we consider a random surface model sampled from the distribution $\mu_{\T_{2L}^d, U}$ of Setting~\ref{ex:torus} (as introduced in Section~\ref{sec:introduction fluctuations}; in particular, we work on the `even' torus $\T_{2L}^d$).
We impose the restrictions $\int\limits_{\mathbb R} e^{-U(x)} \,dx < \infty$ and
$\inf_{\mathbb R} U(x) > -\infty$. We note that every
non-constant convex potential satisfies these
restrictions, but $U$ is not assumed to be convex in this section.

Let $\varphi$ be sampled from the measure
$\mu_{\T_{2L}^d, U}$ of Setting~\ref{ex:torus}. Given a measurable set $S \subseteq [0,\infty)$,
for an arbitrary set of edges $E_0 \subseteq E(\mathbb T_{2L}^d)$
one can consider the following event: the gradients $|\nabla_e\varphi|$
for edges $e \in E_0$ all belong to the set $S$. Our goal is to
obtain an upper bound for the probability of such an event in the form
of $p(S,d,U)^{|E_0|}$, where $p(S, d, U)$ is a certain explicit
expression.  The proof uses reflection positivity in the
form of the chessboard estimate (in a way similar to~\cite{MP15}). In the later application to proving
Theorem~\ref{thm:main}, $S$ is chosen to contain a neighborhood of infinity and a neighborhood of the points where $U''$ vanishes. The neighborhoods are chosen as a function of $U$ in a such a way that $p(S,d,U)$ is small.

\begin{lem}\label{lem:sparsity_estimate}
Let $d \ge 2$ and $L\ge 2$ be integers. Let
$U : \mathbb R \to (-\infty, \infty]$ be a potential
such that $U(-x) = U(x)$ for all $x \in \mathbb R$,
$0 < \int\limits_{\mathbb R} e^{-U(x)} \,dx < \infty$ and
$\inf_{\mathbb R} U(x) > -\infty$.
Let, finally, $\varphi$ be a random surface randomly sampled according
to the measure $\mu_{\T_{2L}^d, U}$. There exist positive numbers
$C(d, U)$ and $c(d)$ (both independent of $L$ and the second also
independent of $U$) such that the inequality
\begin{equation*}
  \Pr \biggl( |\nabla_e\varphi|\in S\text{ for all $e\in E_0$} \biggr)
  \leq p(S,d,U)^{|E_0|}
\end{equation*}
holds for all $E_0 \subseteq E(\T_{2L}^d)$ and all measurable
$S \subseteq [0,\infty)$ if we set, by definition,
\begin{equation*}
  p(S,d,U) := \left( C(d, U) \cdot \int_S e^{-U(x)}\, dx \right)^{c(d)}.
\end{equation*}
\end{lem}

Our use of the chessboard estimate is the main reason for working with periodic boundary conditions (i.e., on $\mathbb T_{2L}^d$). We do not know to obtain an analog of Lemma~\ref{lem:sparsity_estimate} without using the chessboard estimate although it seems reasonable that an analog in Setting~\ref{ex:box} (i.e., on the box $\Lambda_L^d$) should hold. For a step in this direction see~\cite{CaP17}.

The rest of this section is devoted to the proof of
Lemma~\ref{lem:sparsity_estimate}. We therefore extend the use of
the notation introduced in~Lemma~\ref{lem:sparsity_estimate}
to the entire section.

Define
\begin{equation*}
  E_S[\varphi] := \{ e \in E(\T_{2L}^d) : |\nabla_e\varphi| \in S \}.
\end{equation*}
(We use the square brackets to produce a distinguishable notation; the
argument of $E$ accommodated in round brackets is usually a graph.)
Next, for $1\le j\le d$ and $\sigma\in\{0,1\}^d$, define
$E_{j,\sigma}(\T_{2L}^d) \subseteq E(\T_{2L}^d)$ as a collection of edges
of the form $\{(x_1, x_2, \ldots, x_d), (y_1, y_2, \ldots, y_d)\}$, where
\begin{equation*}
  x_i \equiv y_i \equiv \sigma_i \pmod{2}\quad\text{for all $i\neq j$}.
\end{equation*}
In particular, each edge $e \in E_{j,\sigma}(\T_{2L}^d)$ is aligned with
the $j$th coordinate vector. Note that $E_{j,\sigma}(\T_{2L}^d)$ and
$E_{j',\sigma'}(\T_{2L}^d)$ are either disjoint or equal, with equality occurring exactly when $j = j'$ and $\sigma$ equals $\sigma'$ on all but the $j$th coordinate.

The next proposition is known in the literature.
\begin{prop}\label{prop:chessboard}
The inequality
\begin{equation*}
  \Pr \biggl( E_0 \subseteq E_S[\varphi] \biggr)
  \leq \Pr \biggl(E_{j,\sigma}(\T_{2L}^d) \subseteq E_S \biggr)^
  {|E_0| / | E_{j,\sigma}(\T_{2L}^d) |}.
\end{equation*}
holds for each $1\le j\le d$, $\sigma \in \{ 0, 1 \}^d$, all
$E_0 \subseteq E_{j,\sigma}(\T_{2L}^d)$ and all measurable
$S \subseteq [0,\infty)$.
\end{prop}
\begin{proof}
See~\cite[Theorem 5.8]{Bi09} or~\cite[Section~3]{MP15}. The proof is
derived from the chessboard estimate. (See the aforementioned
references also for the details of that technique.)
\end{proof}

For the next proposition the reader may benefit from recalling the
notion of the partition function $Z_{G, V_0, \varphi_0, U}$,
see formula~\eqref{eq:z_def}. In our case, i.e., with $G = \T_{2L}^d$ and
with boundary conditions $(V_0, \varphi_0)$ as in Setting~\ref{ex:torus}
in Section~\ref{sec:introduction}, we shorten the notation to
$Z_{\T_{2L}^d, U}$.

\begin{prop}\label{prop:partition_func_est}
There exists a positive constant $C_3(d, U)$, independent of $L$, such that
\begin{equation*}
  Z_{\T_{2L}^d, U} \geq C_3(d, U)^{1 - | V(\T_{2L}^d) |}.
\end{equation*}
\end{prop}

\begin{proof}
See~\cite[Lemma 3.1]{MP15}.
\end{proof}

\begin{prop}\label{prop:exponential_decay}
There exists a positive constant $C(d, U)$ such that the inequality
\begin{equation}\label{eq:exp_decay}
\Pr \biggl( E_{j,\sigma}(\T_{2L}^d) \subseteq E_S[\varphi] \biggr) \leq
\left( C(d, U) \cdot \int_S e^{-U(t)}\, dt \right)^
{ \frac{1}{2} | E_{j, \sigma}(\T_{2L}^d) |}
\end{equation}
holds for each $1 \leq j \leq d$, each $\sigma \in \{ 0, 1 \}^d$
and every measurable subset $S \subseteq [0, \infty)$.
\end{prop}

\begin{proof}
Let $T$ be an arbitrary tree satisfying the conditions
\begin{itemize}
\item[(a)] $V(T) = V(\T_{2L}^d)$;
\item[(b)] $E(T) \subseteq E(\T_{2L}^d)$;
\item[(c)] $\left| \bigl(E(T) \cap E_{j,\sigma}(\T_{2L}^d)\bigr) \right|
\geq \frac{1}{2} | E_{j, \sigma}(\T_{2L}^d) |$.
\end{itemize}
Then the following inequalities hold.
\begin{multline}\label{eq:exp_decay_chain}
\Pr \biggl( E_{j, \sigma}(\T_{2L}^d) \subseteq E_S[\varphi] \biggr) \leq
\Pr \biggl( E(T) \cap E_{j, \sigma}(\T_{2L}^d) \subseteq
E_S[\varphi] \biggr)  \\
= \frac{1}{Z_{\T_{2L}^d, U}} \cdot \int\limits_{\Omega_{\T_{2L}^d, U}}
\prod\limits_{e \in E(\T_{2L}^d)} \exp(-U(\nabla_e \varphi)) \cdot
\prod\limits_{\substack{e \in E(T) \\ e \in E_{j, \sigma}(\T_{2L}^d)}}
\mathbbm 1(|\nabla_e(\varphi)| \in S) \cdot
\prod\limits_{v \in V(\T_{2L}^d) \setminus \{\mathbf 0 \}}
d\varphi(v) \\
= \frac{1}{Z_{\T_{2L}^d, U}} \cdot \int\limits_{\Omega_{\T_{2L}^d, U}}
\prod\limits_{\substack{e \in E(\T_{2L}^d) \\ e \notin E(T)}}
\exp(-U(\nabla_e \varphi)) \cdot
\prod\limits_{\substack{e \in E(T) \\ e \notin E_{j, \sigma}(\T_{2L}^d)}} \exp(-U(\nabla_e \varphi)) \\
\cdot
\prod\limits_{\substack{e \in E(T) \\ e \in E_{j, \sigma}(\T_{2L}^d)}}
\exp(-U(\nabla_e \varphi))\, \mathbbm 1(|\nabla_e \varphi| \in S) \cdot
\prod\limits_{e \in E(T)} d (\nabla_e \varphi) \\
\leq \frac{1}{Z_{\T_{2L}^d, U}} \cdot
\biggl[ \exp \bigl( -\inf\limits_{t \in \mathbb R} U(t) \bigr) \biggr]^
{|E(\T_{2L}^d)| - | E (T) |} \\
\cdot \left(\int\limits_{\mathbb R} e^{-U(t)} \, dt \right)^
{| E(T) \setminus E_{j, \sigma}(\T_{2L}^d) |} \cdot
\left(2 \int\limits_{S} e^{-U(t)} \, dt \right)^
{| E(T) \cap E_{j, \sigma}(\T_{2L}^d) |}.
\end{multline}
We also note that, by the choice of the tree $T$, one has
\begin{equation}\label{eq:intersection_with_tree}
\left| E(T) \cap E_{j,\sigma}(\T_{2L}^d)\right| \geq
\frac{1}{2} | E_{j, \sigma}(\T_{2L}^d) |.
\end{equation}
Plugging~\eqref{eq:intersection_with_tree} and the result of
Proposition~\ref{prop:partition_func_est}
into~\eqref{eq:exp_decay_chain} indeed yields
the estimate~\eqref{eq:exp_decay}.
\end{proof}

We are ready to finish the proof of Lemma~\ref{lem:sparsity_estimate}.

\begin{proof}[Proof of Lemma~\ref{lem:sparsity_estimate}]
Given $E_0$, choose
\begin{equation*}
  (j_0, \sigma_0) := \argmax\limits_
  {\substack{(j, \sigma): \\ j \in \{1, 2, \ldots, d \}
  \\ \sigma \in \{ 0, 1 \}^d }}
  \left| E_0 \cap E_{j, \sigma}(\T_{2L}^d) \right|.
\end{equation*}
Then
\begin{equation}\label{eq:large_intersection_with_ejs}
  \left| E_0 \cap E_{j_0, \sigma_0}(\T_{2L}^d) \right|
   \geq \frac{| E_0 |}{C_1(d)},
\end{equation}
where, by definition, $C_1(d) := d \cdot 2^{d - 1}$ is the number of
distinct subsets $E_{j, \sigma}(\T_{2L}^d)$.

Let us note that the following inequalities hold:
\begin{multline*}
\Pr \biggl( E_0 \subseteq E_S[\varphi] \biggr) \leq
\Pr \biggl( E_0 \cap E_{j_0, \sigma_0}(\T_{2L}^d) \subseteq
E_S[\varphi] \biggr) \\
\leq \Pr \biggl( E_{j_0, \sigma_0}(\T_{2L}^d) \subseteq
E_S[\varphi] \biggr)^
{| E_0 \cap E_{j_0, \sigma_0}(\T_{2L}^d) | /
| E_{j_0, \sigma_0}(\T_{2L}^d) |} \\
\leq \left( C_2(d, U) \cdot \int\limits_{S} e^{-U(t)}\, dt \right)^
{ \tfrac{1}{2} |E_0 \cap E_{j_0, \sigma_0}(\T_{2L}^d) |  }
\leq \left( C_2(d, U) \cdot \int\limits_{S} e^{-U(t)}\, dt \right)^
{\tfrac{1}{2} | E_0 | / C_1(d) }.
\end{multline*}
Indeed, the first inequality holds by inclusion of the respective events;
the second inequality holds by Proposition~\ref{prop:chessboard};
the third inequality follows from Proposition~\ref{prop:exponential_decay}
(with $C_2$ being the constant from that proposition);
finally, the last inequality is a consequence
of~\eqref{eq:large_intersection_with_ejs}.

We finish the proof of Lemma~\ref{lem:sparsity_estimate} by letting
$C(d, U) := C_2(d, U)$ and $c(d) := \frac{1}{2C_1(d)}$.
\end{proof}

\section{Proofs of the results on random surfaces}\label{sec:proofs of results on random surfaces}
In this section we prove our main results on random surfaces, Theorem~\ref{thm:main} and Theorem~\ref{thm:tail_estimate}.

\subsection{Tail estimate for the potential $U(x) = |x|^p + x^2$, $p>2$}\label{sec:proof of tail estimate theorem}
In this section we prove Theorem~\ref{thm:tail_estimate}.
We will focus on Setting~\ref{ex:box} (the box $\Lambda_L^d$), as the proof for Setting~\ref{ex:torus} is literally the same.

The proof is an application of the quantitative log-concavity result, Theorem~\ref{thm:quantitative log concavity}. Let us first state a corollary of that result in the context of random surfaces. Recall the definitions of the random surface measure $\mu_{G, V_0,\varphi_0, U}$ and the set $\Omega_{G, V_0, \varphi_0}$ from~\eqref{eq:mu_T_n_2_U_measure_def} and~\eqref{eq:omega_def}.

\begin{lem}\label{lem:upper_bound_1}
Let $G$ be a finite connected graph, $V_0$ a proper subset of $V(G)$, $\varphi_0:V_0\to\R$ and $U : \mathbb R \to \mathbb R \cup \{ \infty \}$ be an even convex function which is not everywhere constant. Let $\eta\in\R^{V(G)}$ satisfy $\eta \vert_{V(G) \setminus V_0} \not\equiv 0$. Denote by $\alpha_\eta:\R\to[0,\infty)$ the (log-concave) density function of $\langle \eta, \varphi\rangle$, when $\varphi$ is sampled from the random surface measure $\mu_{G, V_0,\varphi_0, U}$. Define, for $s,r\in\R$ and $t>0$,
\begin{align}
D_{\eta}(s,t) & := \inf_
{\substack{\varphi^+, \varphi^-\in \Omega_{G, V_0, \varphi_0}\\
\langle \eta, \varphi^+\rangle=s+t\\
\langle \eta, \varphi^-\rangle = s-t}}
\sum_{e\in E(G)} \frac{1}{2}\left[U(\nabla_e\varphi^+) +
U(\nabla_e\varphi^-)\right] -
U\Big(\nabla_e\frac{\varphi^+ + \varphi^-}{2}\Big), \label{eq:E_f_def}\\
W(r) & := \inf_{s\in\R} \frac{1}{2}(U(s+r) + U(s-r))-U(s),
\label{eq:W_def} \\
D_{\eta}(t) & := \inf_
{\substack{\psi:V(G) \to \R\\
\psi\equiv 0\text{ on }V_0\\
\langle \eta, \psi \rangle=t}}
\sum_{e\in E(G)} W(\nabla_e\psi).\label{eq:E_f_short_def}
\end{align}
If $s\in\R$, $t>0$ and $\min \{ \alpha_{\eta}(s - t), \alpha_{\eta}(s),
\alpha_{\eta}(s + t) \} > 0$, then
\begin{equation}\label{eq:alpha_f_bound}
\ln \alpha_{\eta}(s) - \frac{1}{2}(\ln \alpha_{\eta}(s + t) +
\ln \alpha_{\eta}(s - t)) \geq D_{\eta}(s, t) \geq D_{\eta}(t).
\end{equation}
\end{lem}
\begin{proof}
In order to use Theorem~\ref{thm:quantitative log concavity} we identify $\Omega_{G, V_0, \varphi_0}$ with $\R^{V(G)\setminus V_0}$ in the canonical fashion (restricting the functions to $V(G)\setminus V_0$).
  We let the random vector $X$ of Theorem~\ref{thm:quantitative log concavity} be sampled from $\mu_{G, V_0,\varphi_0, U}$, noting that this distribution is absolutely continuous with respect to the Lebesgue measure on $\Omega_{G, V_0, \varphi_0}$, with the log-concave density $\exp(-f)$ satisfying
  \begin{equation}\label{eq:f in random surface context}
  f(\varphi) = -\ln(Z_{G,V_0,\varphi_0, U})+\sum_{e\in E(G)}
  U(\nabla_e\varphi).
\end{equation}
We may thus define $D_{\eta,\varphi}(t)$ via the formula~\eqref{eq:E_f_2_def}. It then holds that
\begin{equation*}
  D_\eta(s,t) = \inf_
{\substack{\varphi\in \Omega_{G, V_0, \varphi_0}\\
\langle \eta, \varphi\rangle=s}} D_{\eta, \varphi}(t).
\end{equation*}
Recalling also the definition of $\gamma_{\eta}(D, s, t)$ from~\eqref{eq:gamma_E_s}, our definitions imply that $\gamma_{\eta}(D_\eta(s,t), s,t) = 1$. Thus the inequality
\begin{equation*}
\sqrt{\alpha_{\eta}(s-t)\alpha_{\eta}(s+t)} \le
e^{-D_\eta(s,t)} \alpha_{\eta}(s)
\end{equation*}
holds for every $t>0$ and $s\in\R$ satisfying $\alpha_\eta(s)>0$ by Theorem~\ref{thm:quantitative log concavity}, establishing the first inequality in~\eqref{eq:alpha_f_bound}.

To see that $D_\eta(s,t)\ge D_\eta(t)$, establishing the second inequality in~\eqref{eq:alpha_f_bound}, observe the following. For each $\varphi^+, \varphi^-\in \Omega_{G, V_0, \varphi_0}$ satisfying $\langle \eta, \varphi^+\rangle=s+t$ and $\langle \eta, \varphi^-\rangle = s-t$ we may define $\varphi = \frac{1}{2}(\varphi^+ + \varphi^-)$ and $\psi = \frac{1}{2}(\varphi^+-\varphi^-)$ so that
\begin{equation*}
  \frac{1}{2}\left[U(\nabla_e\varphi^+) +
U(\nabla_e\varphi^-)\right] -
U\Big(\nabla_e\varphi\Big) = \frac{1}{2}(U(\nabla_e\varphi+\nabla_e\psi) + U(\nabla_e\varphi-\nabla_e\psi))-U(\nabla_e\varphi)\ge W(\nabla_e\psi)
\end{equation*}
and $\psi:V(G) \to \R$ satisfies $\psi\equiv 0$ on $V_0$ and $\langle \eta, \psi \rangle=t$.
\end{proof}

We proceed to prove Theorem~\ref{thm:tail_estimate} in Setting~\ref{ex:box}. Fix integers $d\ge 3$, $L\ge 2$ and real $p>2$, $t>1$. Set $U(x) = |x|^p + x^2$. Let $\varphi$ be sampled from the random surface measure $\mu_{\Lambda_{L}^d, U}$. Let $V_0 \subset V(\Lambda_{L}^d)$ be as in Setting~\ref{ex:box} and fix
$v \in V(\Lambda_{L}^d) \setminus V_0$.

Let $\alpha_v$ be the marginal density of $\varphi(v)$. The function $\alpha_v(s)$ is even and log-concave and therefore
$\alpha_v(s)$ is non-strictly increasing for $s \leq 0$ and non-strictly
decreasing for $s \geq 0$.  Combining this observation with Item~\ref{item:logconcave_tail} of
Proposition~\ref{prop:1d_logconcave_properties} gives
\begin{equation*}
	\Pr(|\varphi(v)| > t)
	\leq \Pr(\alpha_v(\varphi(v)) \le \alpha_v(t))
	\leq \frac{\alpha_v(t)}{\alpha_v(0)}.
\end{equation*}
Consequently, it will suffice to prove the inequality
\begin{equation}\label{eq:density_decay}
   \frac{\alpha_v(t)}{\alpha_v(0)} \leq C\exp(-c\, t^{\min\{p, d\}}).
\end{equation}

For convenience, in this section we will use the
comparison operators $\ll$, $\gg$ and $\approx$
meaning, respectively, that the (positive) expression on the left
is smaller than, is greater than, or equals the (positive) expression
on the right up to a positive factor which depends only on $d$ and $p$ and not on any other parameter.

Lemma~\ref{lem:upper_bound_1}, applied to $\alpha_v$ (noting that $\varphi(v) = \langle \eta_v, \varphi \rangle$,
where $\eta_v : V(\Lambda_L^d) \to \mathbb R$ is the indicator function
of the vertex $v$), yields
\begin{equation}\label{eq:alpha v bound}
  \alpha_v(t) = \sqrt{\alpha_v(-t) \alpha_v(t)} \le \exp(-D(t)) \cdot
  \alpha_v(0)
\end{equation}
where, using the abbreviations $\Omega$ for $\Omega_{\Lambda_L^d, V_0, \varphi_0}$ and $E$ for $E(\Lambda_L^d)$,
\begin{align}
  &D(t) = \inf_{\substack{\psi \in \Omega \\ \psi(v) = t}}
  \sum_{e\in E} W(\nabla_e\psi),\label{eq:D(t) def} \\
  &W(r) = \inf_{s\in\R} \frac{1}{2}(U(s+r) + U(s-r)) - U(s)
  \approx |r|^p + r^2.
\end{align}
In the last expression, the approximate equality follows by observing that $\frac{1}{2}((s+r)^2+(s-r)^2) - s^2 = r^2$ for all $s,r$, while $\frac{1}{2}(|s+r|^p+|s-r|^p) - |s|^p$ is $\approx |r|^p$ when $|r|\ge |s|$ and, by considering the second derivative of $|x|^p$, is $\approx r^2|s|^{p-2}$ when $|r|\le |s|$.
Thus, the required inequality~\eqref{eq:density_decay} will follow from~\eqref{eq:alpha v bound} by obtaining a lower bound for the expression
\begin{equation*}
  D^*(t) := \min_{\substack{\psi \in \Omega \\ \psi(v) = t}}
   \sum_{e\in E}
   \bigl( |\nabla_e \psi|^p + (\nabla_e \psi)^{2} \bigr).
\end{equation*}
This will be accomplished by means of the energy estimate of Lemma~\ref{lem:energy_estimate_2}
for $G = \Lambda_L^d$, $a = v$ and an arbitrary vertex $b \in V_0$.
We adopt the notation $m_i(\cdot)$ and $M_i(\cdot)$ of that
lemma. The number of vertices of the graph is
$N := |V(\Lambda_L^d)| = L^d$ and we fix $l$ to be an integer such that
$2^l \leq N < 2^{l + 1}$. It is well-known that
$\Lambda_L^d$ has the isoperimetry of $\mathbb Z^d$ in the sense
of Corollary~\ref{cor:energy_estimate_quadratic}, i.e.,
\begin{equation*}
  M_i(a), M_i(b) \ll \frac{N}{2^i}, \qquad  m_i(a), m_i(b) \gg \left( \frac{N}{2^i} \right)^{\frac{d - 1}{d}},
\end{equation*}
as follows from Lemma~\ref{lem:isoperimetry of cube}.

Lemma~\ref{lem:energy_estimate_2} yields
\begin{multline}\label{eq:energy_kth_power_via_infimum}
  \min_{\substack{\psi \in \Omega \\ \psi(v) = t}}
  \sum_{e\in E} W(\nabla_e\psi) =
  \min_{\substack{\psi \in \Omega \\ \psi(v) = 1}}
   \sum_{e\in E} W \left(t \nabla_e\psi \right) \\
  \geq \inf\limits_{\substack{p_1, \ldots, p_l \geq 0 \\ q_1, \ldots, q_l \geq 0 \\ p_1 + \ldots + p_l + q_1 + \ldots + q_l = 1}}
  \left( \sum\limits_{i = 1}^{l} M_i(a) \cdot W \left(\frac{t p_i m_i(a)}{M_i(a)} \right) +  \sum\limits_{i = 1}^{l} M_i(b) \cdot W \left(\frac{t q_i m_i(b)}{M_i(b)} \right) \right) \\
  \gg \inf\limits_{\substack{p_1, \ldots, p_l \geq 0 \\ p_1 + \ldots + p_l = 1}} \sum\limits_{i = 0}^{l - 1} \left( 2^{i \left( 1 - \tfrac{2}{d} \right)}
  t^2 p_{l - i}^2 + 2^{i \left( 1 - \tfrac{p}{d} \right)} t^{p} p_{l - i}^{p} \right).
\end{multline}

For the rest of the proof we may assume that $l$ is sufficiently large.
Indeed, if $\sigma_l$ denotes the infimum in the rightmost part
of~\eqref{eq:energy_kth_power_via_infimum} and
$f_l(p_1, p_2, \ldots, p_l)$ denotes the expression under that infimum,
we have
\begin{multline*}
\sigma_l = \inf\limits_{\substack{p_1, \ldots, p_l \geq 0
\\ p_1 + \ldots + p_l = 1}} f_l(p_1, p_2, \ldots, p_l)
= \inf\limits_{\substack{p_1, \ldots, p_l \geq 0
\\ p_1 + \ldots + p_l = 1}} f_{l + 1}(0, p_1, p_2, \ldots, p_l)
\\ \geq \inf\limits_{\substack{p_1, \ldots, p_{l+1} \geq 0
\\ p_1 + \ldots + p_{l + 1} = 1}}
f_{l + 1}(p_1, p_2, \ldots, p_{l + 1}) = \sigma_{l + 1}.
\end{multline*}
Consequently, increasing $l$ may only weaken our estimate.

In addition, the fact that the partial derivative of $f_l(p_1, p_2, \ldots, p_l)$ with respect to each $p_i$ equals zero at $p_i=0$ and is strictly positive when $p_i>0$ implies that the infimum is attained for a vector with strictly positive coordinates.

Now, set
\begin{equation*}
  (p_1, p_2, \ldots, p_l) = \argmin \limits_{\substack{p_1, \ldots, p_l \geq 0 \\ p_1 + \ldots + p_l = 1}} \sum\limits_{i = 0}^{l - 1} \left( 2^{i \left( 1 - \tfrac{2}{d} \right)}
  t^2 p_{l - i}^2 + 2^{i \left( 1 - \tfrac{p}{d} \right)} t^{p} p_{l - i}^{p} \right).
\end{equation*}
As explained above we have $p_i>0$ for all $i$. Therefore, there exists a $\beta$ independent of $i$ (a Lagrange multiplier) such that
\begin{equation*}
  2 \cdot 2^{i \left( 1 - \tfrac{2}{d} \right)} t^2 p_{l - i} + p \cdot 2^{i \left( 1 - \tfrac{p}{d} \right)} t^{p} p_{l - i}^{p - 1} = \beta \quad \forall i = 1, 2, \ldots, l.
\end{equation*}
Consequently,
\begin{equation}\label{eq:p ell minus i equality}
  p_{l - i} \approx \min \left(2^{-i \left( 1 - \tfrac{2}{d} \right)} t^{-2} \beta, \quad 2^{-i \tfrac{d - p}{d(p - 1)}}  t^{-\tfrac{p}{p - 1}} \beta^{\tfrac{1}{p - 1}} \right).
\end{equation}

We conclude the proof by considering three cases:

\noindent {\bf Case 1.} $d < p$. Assuming, as we may, that $l$ is sufficiently large, let $0\le i_0\le l-1$ be such that $t^d\le 2^{i_0} < 2t^d$ (recalling that $t>2$). We prove that $p_{l-i_0}\gg1$ whence $D^*(t) \gg t^d$ by~\eqref{eq:energy_kth_power_via_infimum}. Indeed, by~\eqref{eq:p ell minus i equality} and our assumption that $d<p$,
\begin{align*}
  1 = p_1 + \ldots + p_l &\ll  \sum_{i=i_0}^{l-1} 2^{-i \left( 1 - \tfrac{2}{d} \right)} \frac{\beta}{t^2} + \sum_{i=0}^{i_0-1} 2^{i \tfrac{p-d}{d(p - 1)}}  \left(\frac{\beta}{t^p}\right)^{\tfrac{1}{p-1}}\\
  &\ll 2^{-i_0 \left( 1 - \tfrac{2}{d} \right)}\frac{\beta}{t^2} + 2^{i_0 \tfrac{p-d}{d(p - 1)}}  \left(\frac{\beta}{t^p}\right)^{\tfrac{1}{p-1}}\approx \frac{\beta}{t^d} + \left(\frac{\beta}{t^d}\right)^{\tfrac{1}{p-1}}.
\end{align*}
Thus $\beta\gg t^d$ so that $p_{l - i_0}\gg 1$ by~\eqref{eq:p ell minus i equality}.

\noindent {\bf Case 2.} $d = p$. Assuming, as we may, that $l$ is sufficiently large, let $0\le i_0\le l-1$ be such that
\begin{equation*}
  \frac{t^d}{(\ln t)^{d(d-1)/(d-2)}} \le 2^{i_0}< 2\frac{t^d}{(\ln t)^{d(d-1)/(d-2)}}
\end{equation*}
(recalling that $t>2$). We prove that $p_{l-i}\gg \frac{1}{\ln t}$ for $0\le i\le i_0$ whence $D^*(t) \gg \frac{t^d}{(\ln t)^{d-1}}$ by~\eqref{eq:energy_kth_power_via_infimum}. Indeed, by~\eqref{eq:p ell minus i equality} and our assumption that $d=p$,
\begin{equation*}
  1 = p_1 + \ldots + p_l \ll \sum_{i=i_0}^{l-1} 2^{-i \left( 1 - \tfrac{2}{d} \right)} \frac{\beta}{t^2} + \sum_{i=0}^{i_0-1}  \left(\frac{\beta}{t^d}\right)^{\tfrac{1}{d-1}}\ll 2^{-i_0 \left( 1 - \tfrac{2}{d} \right)} \frac{\beta}{t^2} + i_0\left(\frac{\beta}{t^d}\right)^{\tfrac{1}{d-1}}.
\end{equation*}
Thus $\beta\gg \tfrac{t^d}{(\ln t)^{d-1}}$ so that $p_{l - i}\gg \frac{1}{\ln t}$ for $0\le i\le i_0$ by~\eqref{eq:p ell minus i equality}.

\noindent {\bf Case 3.} $d > p$. We prove that $p_l\gg1$ whence $D^*(t) \gg t^{p}$ by~\eqref{eq:energy_kth_power_via_infimum}. Indeed, by~\eqref{eq:p ell minus i equality} and our assumption that $d>p$,
\begin{equation*}
  1 = p_1 + \ldots + p_l \ll \sum_{i=0}^{l-1} 2^{-i \tfrac{d - p}{d(p - 1)}}  \left(\frac{\beta}{t^p}\right)^{\tfrac{1}{p-1}} \ll \left(\frac{\beta}{t^p}\right)^{\tfrac{1}{p-1}}.
\end{equation*}
Thus $\beta\gg t^p$ so that $p_\ell \gg 1$ by~\eqref{eq:p ell minus i equality} (as $p>2$).

Theorem~\ref{thm:tail_estimate} is now proved in Setting~\ref{ex:box} (the box $\Lambda_L^d$). Setting~\ref{ex:torus} is handled in the same way.

\begin{remark}
  Theorem~\ref{thm:tail_estimate} discusses potentials of the form $U(x) = |x|^p + x^2$. It is also natural to consider the family of potentials $U_p(x) = |x|^p$ where we further assume $p\ge1$ to impose convexity. In the case $p>2$ we may define $W_p$ via the recipe~\eqref{eq:W_def} (with respect to $U_p$) and observe that $W_p(r) \approx |r|^p$. With this observation we may follow the calculation in this section and obtain in dimensions $d>p$ the tail probability decay $\P(|\varphi(v)|>t)\le C_p \exp(-c_p t^p)$ (in the setup of Theorem~\ref{thm:tail_estimate}). A more involved approach is required to handle dimensions $d\le p$ or the cases $1\le p<2$ and we do not pursue this direction in this paper.
\end{remark}

\begin{remark}
The above proof of Theorem~\ref{thm:tail_estimate} may be extended to settings with non-zero boundary conditions. We demonstrate this on the $d$-dimensional box $\Lambda_L^d$. Fix integers $d\ge3$, $L\ge 2$ and real $p>2$. Consider the random surface measure $\mu_{\Lambda_L^d, V_0, \varphi_0, U}$ where $U(x)=|x|^p+x^2$, $V_0$ are the vertices of $V(\Lambda_L^d)$ which are adjacent in $\Z^d$ to a vertex outside $V(\Lambda_L^d)$ (as in Setting~\ref{ex:box}) and the boundary condition $\varphi_0:V_0\to\R$ is general. Let $\varphi$ be sampled from $\mu_{\Lambda_L^d, V_0, \varphi_0, U}$ and let $v\in V(\Lambda_L^d)\setminus V_0$. We claim that there exists $C>0$, depending on $p$ and $d$ but not on $L$ or $v$, such that for all real $t>0$,
\begin{equation}\label{eq:concentration around mean}
   \P(|\varphi(v) - \E(\varphi(v))|>2(t+C))\le \exp(-2D(t))
\end{equation}
where $D(t)$ is defined in~\eqref{eq:D(t) def}. From the lower bounds above on $D(t)$ we conclude that~\eqref{eq:concentration around mean} extends Theorem~\ref{thm:tail_estimate} to non-zero boundary conditions, bounding the probability of deviation from $\E(\varphi(v))$ with a bound having the same form as~\eqref{eq:tail probability bound intro}.

Let us prove~\eqref{eq:concentration around mean}. Let $\alpha_v$ be the marginal density of $\varphi(v)$. Lemma~\ref{lem:upper_bound_1}, applied to $\alpha_v$, yields that for each $s\in\R$ and $t>0$,
\begin{equation}\label{eq:alpha v bound in remark}
  \sqrt{\alpha_v(s-t) \alpha_v(s+t)} \le \exp(-D(t)) \cdot
  \alpha_v(s).
\end{equation}
Denote $M:=\sup_{s\in\R}\alpha_v(s)$ and let $s_0$ be such that $\alpha_v(s_0)=M$. Using~\eqref{eq:alpha v bound in remark} with $s=s_0-t$ gives
\begin{equation*}
  \sqrt{\alpha_v(s_0-2t)M}= \sqrt{\alpha_v(s_0-2t)\alpha_v(s_0)}\le \exp(-D(t))\cdot \alpha_v(s_0-t)\le \exp(-D(t))\cdot M.
\end{equation*}
Together with the analogous calculation for $s=s_0+t$ we conclude that
\begin{equation*}
  \max\{\alpha_v(s_0-2t), \alpha_v(s_0+2t)\}\le M\exp(-2D(t)).
\end{equation*}
Thus, by item~\ref{item:logconcave_tail} of Proposition~\ref{prop:1d_logconcave_properties},
\begin{equation}\label{eq:tail decay estimate around s_0}
  \P(|\varphi(v) - s_0|>2t)\le \P(\alpha_v(\varphi(v))\le M\exp(-2D(t)))\le \exp(-2D(t)).
\end{equation}
To deduce~\eqref{eq:concentration around mean}, it remains to note that~\eqref{eq:tail decay estimate around s_0}, together with the fast growth of $D(t)$, imply that $|s_0 - \E(\varphi(v))|\ll 1$.
\end{remark}

\subsection{Proof of Theorem~\ref{thm:main}}\label{sec:proof of main theorem}
In this section we prove Theorem~\ref{thm:main}. We present two proofs: the first based on the quantile Brascamp--Lieb type inequality of Theorem~\ref{thm:quantile Brascamp-Lieb} and the second based on the quantitative log-concavity result of Theorem~\ref{thm:quantitative log concavity}. We will use the definitions of the random surface measure $\mu_{G, V_0,\varphi_0, U}$ and the set $\Omega_{G, V_0, \varphi_0}$, which the reader may recall from~\eqref{eq:mu_T_n_2_U_measure_def} and~\eqref{eq:omega_def}.

The next lemma is an adaptation of the quantile Brascamp--Lieb type inequality, Theorem~\ref{thm:quantile Brascamp-Lieb}, to the random surface setting.

\begin{lem}
\label{lem:quantile Brascamp-Lieb random surfaces}
  There exists a universal constant $C>0$ so that the following holds. Let $G$ be a finite connected graph, $V_0$ a proper subset of $V(G)$, $\varphi_0:V_0\to\R$ and $U : \mathbb R \to \mathbb R \cup \{ \infty \}$ be an even convex function which is not everywhere constant. Let $\eta\in\R^{V(G)}$ satisfy $\eta \vert_{V(G) \setminus V_0} \not\equiv 0$. Let $\varphi$ be sampled from the random surface measure $\mu_{G, V_0,\varphi_0, U}$. Define, for each $\psi\in\Omega_{G, V_0, \varphi_0}$,
  \begin{equation}\label{eq:hessian expression}
    D_{\eta,\psi} := \inf\limits_{\substack
{\substack{\chi:V(G) \to \R\\
\chi\equiv 0\text{ on }V_0 \\
\langle \eta, \chi \rangle = 1}}} \sum_{e\in E(G)} U''(\nabla_e \psi)\left(\nabla_e \chi\right)^2.
  \end{equation}
  Then for each $t>0$,
  \begin{equation}
    \Var(\langle \eta, \varphi\rangle)\le \frac{Ct}{\P\left(D_{\eta,\varphi}\ge \frac{1}{t}\right)^{3}}.
  \end{equation}
  In particular,
  \begin{equation}
    \Var(\langle \eta, \varphi\rangle)\le 8C\Med\left(\frac{1}{D_{\eta,\varphi}}\right).
  \end{equation}
  where $\Med(Y)$ is any median of the random variable $Y$, i.e., a number $t$ satisfying $\P(Y\ge t)\ge \frac{1}{2}$ and $\P(Y\le t)\ge \frac{1}{2}$.
\end{lem}
We note that the term $U''(s)$ appearing in~\eqref{eq:hessian expression} is defined for almost every $s\in\R$ by the convexity of $U$ (in the sense of second-order Taylor expansion; see~\eqref{eq:second order Taylor intro}). The terms corresponding to $e\in E(G|_{V_0})$ in the sum~\eqref{eq:hessian expression} can be omitted due to having $\nabla_e \chi=0$. The other terms are almost surely well defined when $\varphi$ is substituted for $\psi$, so that $D_{\eta,\varphi}$ is well defined.
\begin{proof}
As in the previous proof, identify $\Omega_{G, V_0, \varphi_0}$ with $\R^{V(G)\setminus V_0}$ in the canonical fashion (restricting the functions to $V(G)\setminus V_0$), and note that the distribution of $\varphi$ has the log-concave density $\exp(-f)$ given by~\eqref{eq:f in random surface context}. Noting that the statement of the lemma is unaffected by changes to the coordinates of $\eta$ on $V_0$, we also identify $\eta$ with its restriction to $V\setminus V_0$ when needed. With these in mind, the lemma follows from Theorem~\ref{thm:quantile Brascamp-Lieb}, when substituting $\varphi$ for $X$, after noting that
\begin{equation}\label{eq:D chi equality}
  D_{\eta,\psi} = \frac{1}{\langle \eta, (\Hess f)(\psi)^{-1} \eta \rangle}
\end{equation}
for the set of $\psi\in\Omega_{G, V_0, \varphi_0}$ for which $D_{\eta,\psi}$ is defined by~\eqref{eq:hessian expression}.
To see the equality, recall first the variational principle (see~\eqref{eq:quadratic form variational principle}),
\begin{equation*}
  \frac{1}{\langle \eta, (\Hess f)(\psi)^{-1} \eta \rangle} = \inf\limits_{\substack
{\substack{\chi:V(G) \to \R\\
\chi\equiv 0\text{ on }V_0 \\
\langle \eta, \chi \rangle = 1}}} \langle \chi,\Hess(f)(\psi)\chi\rangle
\end{equation*}
where, again, $\chi$ is identified with its restriction to $V(G)\setminus V_0$, and the left-hand side is interpreted as $0$ when $\eta$ lies in the kernel of $(\Hess f)(\psi)$. Equality~\eqref{eq:D chi equality} then follows by consideration of the formula~\eqref{eq:f in random surface context} for $f$.
\end{proof}

We also adapt the quantitative log-concavity results, Theorem~\ref{thm:quantitative log concavity} and Lemma~\ref{lem:var_via_logconcavity}, to the random surface setting.
\begin{lem}\label{lem:quantitative log concavity random surfaces}
Let $G$ be a finite connected graph, $V_0$ a proper subset of $V(G)$, $\varphi_0:V_0\to\R$ and $U : \mathbb R \to \mathbb R \cup \{ \infty \}$ be an even convex function which is not everywhere constant. Let $\eta\in\R^{V(G)}$ satisfy $\eta \vert_{V(G) \setminus V_0} \not\equiv 0$. Let $\varphi$ be sampled from the random surface measure $\mu_{G, V_0,\varphi_0, U}$ and set $\alpha_\eta:\R\to[0,\infty)$ to be the (log-concave) density function of $\langle \eta, \varphi\rangle$. Define, for $\psi\in\Omega_{G, V_0, \varphi_0}$ and $t>0$,
\begin{equation}\label{eq:energy in quantitative log concavity random surfaces}
D_{\eta, \psi}(t) := \inf_{\substack{\psi^+, \psi^- \in \Omega_{G, V_0, \varphi_0}\\
\psi^+ +\psi^- = 2 \psi \\
\langle \eta, \psi^+ - \psi^- \rangle=2t}}
\sum_{e\in E(G)} \frac{1}{2}\left[U(\nabla_e \psi^+) +
U(\nabla_e \psi^-)\right] -
U\Big( \nabla_e \psi \Big).
\end{equation}
Define further, for each $D\ge0$, $t>0$ and $s\in\R$ satisfying that $\alpha_\eta(s)>0$
\begin{equation}\label{eq:gamma_E_s 2}
\gamma_{\eta}(D, s, t) := \Pr ( D_{\eta, \varphi}(t) \ge D \mathbin{\vert}
\langle \eta, \varphi \rangle = s).
\end{equation}
Then the inequality
\begin{equation}\label{eq:alpha_f_2_bound random surface}
\sqrt{\alpha_{\eta}(s-t)\alpha_{\eta}(s+t)} \le
\bigl(1-\gamma_{\eta}(D, s, t) (1- e^{-D}) \bigr) \cdot \alpha_{\eta}(s)
\end{equation}
holds for every $D,t$ and $s$ as above. In particular, if $\P(D_{\eta, \varphi}(t) \ge D)\ge \frac{3}{4}$ for some $D, t>0$ then
\begin{equation}\label{eq:Markov consequence random surface}
\P\left( \sqrt{\alpha(\langle \eta, \varphi \rangle + t)
 \alpha(\langle \eta, \varphi \rangle - t)}
\leq \Big(1 - \frac{1}{2}\left(1-e^{-D}\right)\Big) \alpha(\langle \eta, \varphi \rangle) \right) \ge \frac{1}{2}
\end{equation}
and consequently
\begin{equation}\label{eq:variance bound random surface}
  \Var(\langle \eta, \varphi \rangle)\le \left(\frac{C t}{1-e^{-D}}\right)^2
\end{equation}
for an absolute constant $C>0$.
\end{lem}
\begin{proof}
  The proof follows the same lines as the proof of Lemma~\ref{lem:upper_bound_1}. Identify $\Omega_{G, V_0, \varphi_0}$ with $\R^{V(G)\setminus V_0}$ in the canonical fashion (restricting the functions to $V(G)\setminus V_0$).
  Noting that the distribution of $\varphi$ is absolutely continuous with respect to the Lebesgue measure on $\Omega_{G, V_0, \varphi_0}$, with the log-concave density $\exp(-f)$ given by~\eqref{eq:f in random surface context}. With these definitions, $D_{\eta,\psi}(t)$ coincides with the formula~\eqref{eq:E_f_2_def} when substituting $\psi$ for $x$ and replacing the $\eta$ of the lemma by its restriction to $V(G)\setminus V_0$. Similarly, $\gamma_\eta(D,s,t)$ defined in~\eqref{eq:gamma_E_s 2} coincides with~\eqref{eq:gamma_E_s} when substituting $\varphi$ for $X$. The inequalities~\eqref{eq:alpha_f_2_bound random surface} and~\eqref{eq:Markov consequence random surface} thus follow from Theorem~\ref{thm:quantitative log concavity} and the conclusion~\eqref{eq:variance bound random surface} follows from Lemma~\ref{lem:var_via_logconcavity}.
\end{proof}

\subsubsection{The key lemma} We now specialize to the setting of Theorem~\ref{thm:main} and state the key lemma for its proof.

\begin{lem}\label{lem:key lemma}
Suppose that $U:\R\to(-\infty,\infty]$ is such that
$U(x)=U(-x)$ for all $x$, and, in addition, the following assumption
is satisfied:
\begin{equation}
  \text{$U$ is convex and $U''(x)>0$ Lebesgue almost-everywhere (a.e.) on $\{x\colon U(x)<\infty\}$}.
\end{equation}
Let $d\ge 2$ and $L\ge 2$ be integers and let $v\in V(\T_{2L}^d)\setminus\{\zero\}$.
Define
\begin{enumerate}
  \item (Fluctuation growth): For $R>0$,
  \begin{equation}\label{eq:fluctuation growth}
  \tau_d(R) := \begin{cases}
    \sqrt{\ln(R+1)}&d=2,\\
    1&d\ge 3.
  \end{cases}
\end{equation}
  \item (Effective conductance of a subgraph): For a set of edges $E\subseteq E(\T_{2L}^d)$,
\begin{equation}\label{eq:effective conductance}
    D_{E,v} := \inf\limits_{\substack
{\substack{\chi:V(\T_{2L}^d) \to \R\\
\chi(\zero)= 0\\
\chi(v) = 1}}}\; \sum_{e\in E} \left(\nabla_e \chi\right)^2.
  \end{equation}
  \item (Second-order ratio of $U$ at $s$): For $s\in\R$,
  \begin{equation}\label{eq:second order ratio}
    \delta_U(s) := \inf\limits_{t\in(0,\infty)} \frac{U(s+t) + U(s-t) - 2U(s)}{\min\{t^2, 1\}}.
  \end{equation}
  \item (Subgraph of edges with large second-order ratio): For $\psi:V(\T_{2L}^d)\to\R$ and $\delta>0$,
    \begin{equation}\label{eq:good edges}
      E(\psi, \delta) = \{e\in E(G)\colon \delta_U(\nabla_e\psi)\ge \delta\}.
    \end{equation}
\end{enumerate}
Then there exist $\delta_0,c>0$ depending only on $d$ and $U$ (but not on $L$ and $v$) such that when $\varphi$ is randomly sampled from $\mu_{\T_{2L}^d, U}$,
\begin{equation}\label{eq:energy on good subgraph prob}
  \P\left(D_{E(\varphi, \delta_0), v}\ge \frac{c}{\tau_d(\|v\|_1)^2}\right)\ge \frac{3}{4}.
\end{equation}
\end{lem}

\subsubsection{Deduction of Theorem~\ref{thm:main}
from Lemma~\ref{lem:key lemma}} We explain here how Theorem~\ref{thm:main} follows from the key lemma using either the quantile Brascamp--Lieb result of Lemma~\ref{lem:quantile Brascamp-Lieb random surfaces} or the quantitative log-concavity result of Lemma~\ref{lem:quantitative log concavity random surfaces}. Let $\Omega = \{\psi:V(\T_{2L}^d) \to \R\colon \psi(\zero)=0\}$.

Let $\eta_v:V(\T_{2L}^d) \to \R$ be defined by $\eta_v(v) = 1$ and $\eta_v(w)=0$ for $w\in V(\T_{2L}^d)\setminus\{v\}$, so that $\langle \eta_v, \psi \rangle = \psi(v)$ for $\psi\in\Omega$. We apply the aforementioned results with $\eta = \eta_v$.

\begin{proof}[Proof of Theorem~\ref{thm:main} using the quantile Brascamp--Lieb approach] Observe that if $U''(s)$ is defined at an $s\in\R$ then
\begin{equation}
  U''(s)\ge \delta_U(s).
\end{equation}
Thus, for $\psi\in\Omega$, with the definition of $D_{\eta, \psi}$ from~\eqref{eq:hessian expression}, the definition of $D_{E,v}$ from~\eqref{eq:effective conductance} and the definition of $E(\psi,\delta)$ from~\eqref{eq:good edges},
\begin{equation}
  D_{\eta_v,\psi}\ge \delta\cdot D_{E(\psi,\delta), v}.
\end{equation}
As the conclusion~\eqref{eq:energy on good subgraph prob} of the key lemma implies that when $\varphi$ is randomly sampled from $\mu_{\T_{2L}^d, U}$ then
\begin{equation*}
  \Med\left(\frac{1}{D_{\eta_v,\varphi}}\right)\le \frac{\tau_d(\|v\|_1)^2}{c}
\end{equation*}
for at least one median, we conclude from Lemma~\ref{lem:quantile Brascamp-Lieb random surfaces} that
\begin{equation*}
  \Var(\varphi(v))\le C' \tau_d(\|v\|_1)^2
\end{equation*}
for an absolute constant $C'>0$, as required.
\end{proof}

\begin{proof}[Proof of Theorem~\ref{thm:main} using the quantitative log-concavity approach] We aim to give a lower bound
for the quantity $D_{\eta, \psi}(t)$ of \eqref{eq:energy in quantitative log concavity random surfaces} (with $\eta = \eta_v$) in terms of the quantity $D_{E,v}$ of~\eqref{eq:effective conductance} and the set of edges $E(\psi,\delta)$ of~\eqref{eq:good edges}. To this end, we write for $\psi\in\Omega$ and $t,\delta>0$,
\begin{equation}\label{eq:D eta v psi estimate1}
\begin{split}
D_{\eta_v, \psi}(t) &\stackrel{(a)}{=} \inf_{\substack{\psi^+, \psi^- \in \Omega\\
\psi^+ +\psi^- = 2 \psi \\
\psi^+(v) - \psi^-(v)=2t}}
\sum_{e\in E(G)} \frac{1}{2}\left[U(\nabla_e \psi^+) +
U(\nabla_e \psi^-)\right] -
U\Big( \nabla_e \psi \Big)\\
&\stackrel{(b)}{\ge} \inf_{\substack{\psi^+, \psi^- \in \Omega\\
\psi^+ +\psi^- = 2 \psi \\
\psi^+(v) - \psi^-(v)=2t}}
\sum_{e\in E(\psi,\delta)} \frac{1}{2}\left[U(\nabla_e \psi^+) +
U(\nabla_e \psi^-)\right] -
U\Big( \nabla_e \psi \Big)\\
&\stackrel{(c)}{=} \inf_{\substack{\chi \in \Omega\\
\chi(v)=1}}\,
\sum_{e\in E(\psi,\delta)} \frac{1}{2}\left[U(\nabla_e \psi + t\nabla_e\chi) +
U(\nabla_e \psi - t\nabla_e\chi)\right] -
U\Big( \nabla_e \psi \Big)\\
&\stackrel{(d)}{\ge} \inf_{\substack{\chi \in \Omega\\
\chi(v)=1}}\,
\frac{\delta}{2}\sum_{e\in E(\psi,\delta)} \min \left\{t^2(\nabla_e\chi)^2, 1 \right\}\\
&\stackrel{(e)}{\ge} \inf_{\substack{\chi \in \Omega\\
\chi(v)=1}}\,
\frac{\delta}{2}\min\left\{t^2 \sum_{e\in E(\psi,\delta)} (\nabla_e\chi)^2, 1 \right\} \stackrel{(f)}{=} \frac{\delta}{2} \min\left\{t^2 D_{E(\psi,\delta), v},1\right\}
\end{split}
\end{equation}
where the inequality (b) uses the fact that $\frac{1}{2}(U(s+t) + U(s-t)) - U(s)\ge 0$ for all $s,t\in\R$ by the convexity of $U$, the equality (c) follows by defining $\chi = \tfrac{\psi^+ - \psi^-}{2t}$ and the inequality (d) holds by the definition of $E(\psi,\delta)$ (see~\eqref{eq:good edges}). Let $\delta_0, c>0$ be as in the key lemma, Lemma~\ref{lem:key lemma}. Thus, the conclusion~\eqref{eq:energy on good subgraph prob} of the key lemma implies that when $\varphi$ is randomly sampled from $\mu_{\T_{2L}^d, U}$ then
\begin{equation}
  \P\left(D_{\eta_v, \varphi}(\tau_d(\|v\|_1))\ge \frac{\delta}{2}\min\{c,1\}\right)\ge \frac{3}{4}.
\end{equation}
As this is the sufficient condition for the variance bound~\eqref{eq:variance bound random surface} stated in the quantitative log-concavity result of Lemma~\ref{lem:quantitative log concavity random surfaces} we conclude that
\begin{equation}
  \Var(\varphi(v))\le C' \tau_d(\|v\|_1)^2
\end{equation}
for an absolute constant $C'>0$, as required.
\end{proof}

\subsubsection{Proof of the key lemma}\label{sec:proof of prob bound on energy lemma}

\begin{proof}[Proof of Lemma~\ref{lem:key lemma}]
We rely on Lemma~\ref{lem:isoperimetry of cube after percolation}
to introduce an auxiliary constant $p_0\in (0,1)$. In the notation
of that lemma, let $p_0$ be chosen so close to 1
as to satisfy the inequality $q(p_0)\ge \frac{11}{12}$. From now on,
let also $R := \| v \|_{\infty}$.

Assume that, with some choice of $\delta_0 > 0$ the inequality
\begin{equation}\label{eq:gradient_sparsity_assumption}
\Pr\left( E(\varphi, \delta_0) \cap E' = \varnothing \right)
\leq (1 - p_0)^{|E'|}
\end{equation}
holds for every $E' \subseteq E(\T^d_{2L})$. Considering the
embedding $\Lambda_{R + 1}^d \hookrightarrow \T_{2L}^d$ such
that both $\zero$ and $v$ are vertices of the embedded graph,
denote
\begin{equation*}
\Lambda_{R+1, p_0}^d :=
\left(V(\Lambda_{R + 1}^d),
E(\Lambda_{R + 1}^d) \cap E(\varphi, \delta_0) \right).
\end{equation*}
The assumption~\eqref{eq:gradient_sparsity_assumption} implies that
the graph $\Lambda_{R+1, p_0}^d$ satisfies the condition of
Lemma~\ref{lem:isoperimetry of cube after percolation} with $p = p_0$.
Consequently, writing
\begin{align*}
\mathcal E_0 & := \{
\text{event (a) of
Lemma~\ref{lem:isoperimetry of cube after percolation}
for $\Lambda_{R+1, p_0}^d$ with $a = \zero$ and $b = v$ occurs}
\},\\
\mathcal E_1 & := \{
\text{event (b) of
Lemma~\ref{lem:isoperimetry of cube after percolation}
for $\Lambda_{R+1, p_0}^d$ with $a = \zero$ and $b = v$ occurs}
\},\\
\mathcal E_2 & := \{
\text{event (b) of
Lemma~\ref{lem:isoperimetry of cube after percolation}
for $\Lambda_{R+1, p_0}^d$ with $a = v$ and $b = \zero$ occurs}
\},
\end{align*}
we get that $\Pr(\mathcal E_0 \cap \mathcal E_1 \cap \mathcal E_2)
\geq \frac{3}{4}$.

As in Lemma~\ref{lem:isoperimetry of cube after percolation},
denote by $G_u$ the conected component of the vertex $u$ in
$\Lambda_{R+1, p_0}^d$. By definition of the event
$\mathcal E_0$ it follows that $G_{\zero} = G_v$ whenever this
event occurs. In addition, by definition
of the events $\mathcal E_1$ and $\mathcal E_2$, it holds that,
whenever the event $\mathcal E_0 \cap \mathcal E_1 \cap \mathcal E_2$
occurs, the graph $G_{\zero} = G_v$ satisfies the condition of
Corollary~\ref{cor:energy_estimate_quadratic}. As a result,
whenever the event $\mathcal E_0 \cap \mathcal E_1 \cap \mathcal E_2$
occurs, we have
\begin{equation*}
D_{E(\varphi, \delta_0), v} \stackrel{(a)}= \inf\limits_{\substack
{\substack{\chi:V(\T_{2L}^d) \to \R\\
\chi(\zero)= 0\\
\chi(v) = 1}}}\; \sum_{e\in E(\varphi, \delta_0)}
\left(\nabla_e \chi\right)^2
\stackrel{(b)}\geq \inf\limits_{\substack
{\substack{\chi:V(\T_{2L}^d) \to \R\\
\chi(\zero)= 0\\
\chi(v) = 1}}}\; \sum_{e\in E(G_{\zero})}
\left(\nabla_e \chi\right)^2  \stackrel{(c)}\geq
\frac{c}{\tau_d(\|v\|_1)^2},
\end{equation*}
where the inequality (c) follows from
Corollary~\ref{cor:energy_estimate_quadratic}. Hence
\begin{equation*}
\P\left(D_{E(\varphi, \delta_0), v}\ge \frac{c}{\tau_d(\|v\|_1)^2}\right)\geq \Pr(\mathcal E_0 \cap \mathcal E_1 \cap \mathcal E_2)
\geq \frac{3}{4},
\end{equation*}
and~\eqref{eq:energy on good subgraph prob} indeed follows
from~\eqref{eq:gradient_sparsity_assumption}.

Therefore, in order to finish the proof, it is sufficient to choose
$\delta_0 > 0$ so that~\eqref{eq:gradient_sparsity_assumption} holds.
The remaining part of the argument pursues that goal.

Denote
\begin{equation*}
S_U(\delta) := \{s \in \R : U(s) < \infty,\; \delta_U(s) < \delta \}.
\end{equation*}
The set $S_U(\delta)$ is $0$-symmetric; for convenience,
we also denote $S_U^+(\delta) := S_U(\delta) \cap [0, \infty)$.
It is also straightforward that $\delta_U(s) > 0$ if $U''(s) > 0$,
which implies that the identity
\begin{equation*}
\lim\limits_{\delta \searrow 0} \mathbbm{1}
\{s \in S_U(\delta) \} = 0
\end{equation*}
holds Lebesgue-a.e. on the set $\{ s : U(s) < \infty\}$. By the
dominated convergence theorem, we get
\begin{equation*}
\lim\limits_{\delta \searrow 0}
\int\limits_{S_U^+(\delta)} e^{-U(t)} \, dt =
\lim\limits_{\delta \searrow 0} \int\limits_{0}^{\infty} e^{-U(t)}
\mathbbm{1} \{ t \in S_U^+(\delta) \} \, dt = 0.
\end{equation*}
Consequently, with $p(S, d, U)$ defined as in
Lemma~\ref{lem:sparsity_estimate}, we have
\begin{equation*}
\lim\limits_{\delta \searrow 0} p(S_U^+(\delta), d, U) =
\lim\limits_{\delta \searrow 0}
\left( C(d, U) \cdot \int\limits_{S_U^+(\delta)}
e^{-U(x)}\, dx \right)^{c(d)} = 0.
\end{equation*}
Thus there exists $\delta_0 > 0$ depending only on $d$ and $U$ such that
$p(S_U^+(\delta_0), d, U) < 1 - p_0$. For this choice of $\delta_0$
the property~\eqref{eq:gradient_sparsity_assumption} follows directly
from Lemma~\ref{lem:sparsity_estimate}, which concludes the proof.
\end{proof}

\section*{Acknowledgements}
The work of AM was supported in part by Israel Science Foundation grant 861/15, the European Research Council starting grant 678520 (LocalOrder) and the Russian Science Foundation grant 20-41-09009. The work of RP was supported in part by Israel Science Foundation grants 861/15 and 1971/19 and by the European Research Council starting grant 678520 (LocalOrder).
We are grateful to Gady Kozma for letting us know of the work~\cite{BK05} and for sharing with us the question of understanding the typical maximal value of the random surface with potential $U(x) = x^4 + x^2$ on the three-dimensional grid $V(\Lambda_L^d):=\{1,\ldots, L\}^3$. We thank Ronen Eldan, Emanuel Milman and Sasha Sodin for fruitful discussions of the presented results and related questions. We thank Jean-Dominique Deuschel for a discussion of the results of~\cite{DG00}. Lastly, we thank Shangjie Yang and an anonymous referee for helpful comments on the manuscript.

\end{document}